\documentclass[11pt]{article}
\usepackage{fullpage, amsmath, amsthm, amssymb, setspace, color,graphicx}
\usepackage{epstopdf}

\newtheorem{theorem}{Theorem}
\newtheorem{lemma}[theorem]{Lemma}
\newtheorem{definition}[theorem]{Definition}

\newtheorem{proposition}[theorem]{Proposition}
\newtheorem{conjecture}[theorem]{Conjecture}

\numberwithin{theorem}{section}

\newtheorem*{lem:maingeneral}{Lemma \ref{LemmaMainGeneral}}
\newtheorem*{thm:prggeneral}{Theorem \ref{TheoremPRGgeneral}}
\newtheorem*{lem:brry}{Lemma \ref{LemmaBRRY}}

\newcommand{\ex}[2]{\underset{#1}{\mathbb{E}}\left[ #2 \right]}
\newcommand{\pr}[2]{\underset{#1}{\mathbb{P}}\left[ #2 \right]}
\newcommand{\norm}[1]{\left|\left| #1 \right|\right|}
\newcommand{\frob}[1]{\norm{#1}_\text{Fr}}
\newcommand{\omitted}[1]{}

\newcommand{\RL}{\textrm{RL}}
\newcommand{\Lspace}{\textrm{L}}
\newcommand{\poly}{\mathrm{poly}}

\newcommand{\eps}{\varepsilon}
\newcommand{\zo}{\{0,1\}}
\newcommand{\tO}{\tilde{O}}

\newcommand{\Z}{\mathbb{Z}}

\newcommand{\Set}{\mathrm{Select}}
\newcommand{\RR}{\mathbb{R}}

\title{Pseudorandomness and Fourier Growth Bounds \\ for Width 3 Branching Programs}
\author{Thomas Steinke\thanks{School of Engineering and Applied Sciences, Harvard University, 33 Oxford Street, Cambridge MA. Supported by NSF grant CCF-1116616 and the Lord Rutherford Memorial Research Fellowship.}\\\small\texttt{tsteinke@seas.harvard.edu} \and Salil Vadhan\thanks{School of Engineering and Applied Sciences, Harvard University, 33 Oxford Street, Cambridge MA.  Supported in part by NSF grant CCF-1116616, US-Israel BSF grant 2010196, and a Simons Investigator Award.} 
\\\small\texttt{salil@seas.harvard.edu}
\and Andrew Wan\thanks{Simons Institute for the Theory of Computing, UC Berkeley. Part of this work was completed while at Harvard University.}\\\small\texttt{atw12@seas.harvard.edu}}

\date{27 May 2014} 

\begin{document}

\begin{titlepage}

\maketitle
\thispagestyle{empty}

\begin{abstract}
We present an explicit pseudorandom generator for oblivious, read-once, width-$3$ branching programs, which can read their input bits in any order.  The generator has seed length $\tilde{O}( \log^3 n ).$ 
The previously best known seed length for this model is $n^{1/2+o(1)}$ due to Impagliazzo, Meka, and Zuckerman (FOCS '12).  Our work generalizes a recent result of Reingold, Steinke, and Vadhan (RANDOM '13) for \textit{permutation} branching programs.  The main technical novelty underlying our generator is a new bound on the Fourier growth of width-3, oblivious, read-once branching programs. Specifically, we show that for any $f:\zo^n\rightarrow \zo$ computed by such a branching program, and  $k\in [n],$ 
$$\sum_{s\subseteq [n]: |s|=k} \left| \hat{f}[s] \right| \leq n^2 \cdot (O(\log n))^k,$$
where $\widehat{f}[s] = \ex{U}{f[U] \cdot (-1)^{s \cdot U}}$ is the standard Fourier transform over $\Z_2^n$. The base $O(\log n)$ of the Fourier growth is tight up to a factor of $\log \log n$.
\end{abstract}
\end{titlepage}

\section{Introduction} \label{SectionIntroduction}
\subsection{Pseudorandom Generators for Space-Bounded Computation}
A major open problem in the theory of pseudorandomness is to construct an ``optimal'' pseudorandom generator for space-bounded computation.  That is,
we want an
explicit algorithm that stretches a uniformly random seed
of length $O(\log n)$ to $n$ bits that cannot be distinguished from uniform by any $O(\log n)$-space algorithm (which receives the pseudorandom bits one at a time, in a streaming fashion, and may be nonuniform).
Such a generator would imply that every randomized algorithm can be derandomized with only a constant-factor increase in space ($\RL=\Lspace$), and would also have a variety of other applications, such as
in streaming
algorithms~\cite{Indyk},
deterministic dimension reduction and SDP rounding~\cite{Sivakumar,EngebretsonInOd01}, hashing~\cite{CReingoldSW}, hardness
amplification~\cite{HealyVaVi}, almost $k$-wise independent
permutations~\cite{KaplanNaRe}, and cryptographic pseudorandom generator constructions~\cite{HaitnerHaRe06}.

To construct a pseudorandom generator for space-bounded algorithms using space $s$, it suffices to construct a generator that is pseudorandom against ordered branching programs of width $2^s$.      
A branching program\footnote{In this work and the definition we give here, we consider read-once, oblivious branching programs, and refer to them simply as branching programs for brevity.} $B$ is a non-uniform model of space-bounded computation that reads one input bit at a time, maintaining a state in $[w] = \{1,\dots,w\}$, where $w$ is called the width of $B$.  At each time step $i=1,\ldots,n,$ $B$ can read a different input bit $x_{\pi(i)}$ (for some permutation $\pi$) and uses a different state transition function $T_i:[w]\times \zo\to [w]$.  It is often useful to think of a branching program as a directed acyclic graph consisting of $n+1$ layers of $w$ vertices each, where the $i^\text{th}$ layer corresponds to the state at time $i$. The transition function defines a bipartite graph between consecutive layers, where we connect state $s$ in layer $i-1$ to states $T_i(s,0)$ and $T_i(s,1)$ in layer $i$ (labeling those edges 0 and 1, respectively).  
Most previous constructions of pseudorandom generators for space-bounded computations consider \textit{ordered} branching programs, where the input bits are read in order -- that is, $\pi(i)=i$.

The classic work of Nisan~\cite{Nisan} gave a generator with seed length $O(\log ^2 n)$ that is pseudorandom against ordered branching programs of polynomial width.  Despite intensive study, this is the best known seed length for ordered branching programs even of width 3, but a variety of works have shown improvements for restricted classes such as 
branching programs of width 2  \cite{SZ,BogdanovDvVeYe09}, and regular or permutation branching programs (of constant width)~\cite{BRRY,BrodyVerbin,KNP,De,Steinke12}.  For width 3, hitting set generators (a relaxation of pseudorandom generators) have been constructed \cite{SimaZak, GMRTV}.
The vast majority of these works are based on Nisan's original generator or its variants by Impagliazzo, Nisan, and Wigderson~\cite{INW} and Nisan and Zuckerman~\cite{NZ}, and adhere to a paradigm that seems unlikely to yield generators against general logspace computations with seed length better than $\log^{1.99} n$ (see \cite{BrodyVerbin}).

All known analyses of Nisan's generator and its variants rely on the order in which the output bits are fed to the branching program (given by the permutation $\pi$). The search for new ideas leads us to ask: Can we construct a pseudorandom generator whose analysis does not depend on the order in which the bits are read? A recent line of work \cite{BPW,IMZ,ReingoldSteinkeVadhan2013} has constructed pseudorandom generators for unordered branching programs (where the bits are fed to the branching program in an arbitrary, fixed order); however, none match both the seed length and generality of Nisan's result. 
For 
unordered branching programs of length $n$ and width $w$, Impagliazzo, Meka, and  Zuckerman \cite{IMZ} give
seed length $O((nw)^{1/2+o(1)})$ 
improving on the linear seed length $(1-\Omega(1))\cdot n$ of Bogdanov, Papakonstantinou, and Wan \cite{BPW}.\footnote{A generator with seed length $\tilde{O}(\sqrt{n} \log w)$ is given in \cite{ReingoldSteinkeVadhan2013}.  The generator in \cite{IMZ} also extends to branching programs that read their inputs more than once and in an adaptively chosen order, which is more general than the model we consider.  } Reingold, Steinke, and Vadhan \cite{ReingoldSteinkeVadhan2013} achieve seed length $O(w^2 \log^2 n)$ for the restricted class of \textit{permutation} branching programs, in which $T_i(\cdot, b)$ is a permutation on $[w]$ for all $i \in [n]$ and $b \in \{0,1\}$. 

Recently, a new approach for constructing pseudorandom generators has been suggested in the work of Gopalan et al.~\cite{GMRTV}; they constructed pseudorandom generators for read-once CNF formulas and combinatorial rectangles, and hitting set generators for width-3 branching programs, all having seed length $\tilde{O}(\log n)$ (even for polynomially small error).  Their basic generator (e.g. for read-once CNF formulas) works by pseudorandomly partitioning the bits into several groups and assigning the bits in each group using a small-bias generator~\cite{NaorNa93}. 
A key insight in their analysis is that the small-bias generator only needs to fool the function ``on average,'' where the average is taken over the possible assignments to subsequent groups, which is a weaker 
requirement than fooling the original function or even a random restriction of the original function. 
(For a more precise explanation, see Section \ref{SubSectionPRrestriction}.)

The analysis of Gopalan et al.~\cite{GMRTV} does not rely on the order in which the output bits are read, and 
the previously mentioned work by Reingold, Steinke, and Vadhan \cite{ReingoldSteinkeVadhan2013} uses Fourier analysis of branching programs to show that the generator of Gopalan et al.~fools unordered permutation branching programs. 

In this work we further develop Fourier analysis of branching programs and show that the pseudorandom generator of Gopalan et al.~with seed length $\tilde{O}(\log^6 n)$ fools width-3 branching programs:    
\begin{theorem}[Main Result] \label{thm:main-intro}
There is an explicit pseudorandom generator $G : \zo^{O(\log^3 n \cdot \log \log n)} \rightarrow \zo^n$ fooling oblivious, read-once (but unordered),  branching programs of width $3$ and length $n$.
\end{theorem}
The previous best seed length for this model is the aforementioned length of $O(n^{1/2+o(1)})$ given in \cite{IMZ}.
The construction of the generator in Theorem \ref{thm:main-intro} is essentially the same as the generator of Gopalan et al.~\cite{GMRTV} for read-once CNF formulas, which was used by Reingold et al.~\cite{ReingoldSteinkeVadhan2013} for permutation branching programs.  In our analysis, we give a new bound on the Fourier mass of 
width-3 branching programs. 

\subsection{Fourier Growth of Branching Programs}
For a function $f : \zo^n\rightarrow \RR$, let
$\widehat{f}[s] = \ex{U}{f[U] \cdot (-1)^{s\cdot U}}$ be the standard Fourier transform over $\Z_2^n$,  where $U$ is a random variable distributed uniformly over $\zo^n$ and $s\subseteq [n]$ or, equivalently, $s \in \{0,1\}^n$.  The Fourier mass of $f$ (also called the spectral norm of $f$),
defined as $L(f):=\sum_{s\neq \emptyset} |\hat{f}[s]|$, is a fundamental measure of complexity for Boolean functions (e.g., see \cite{green2008boolean}), and its study has applications to learning theory \cite{kushilevitz1993learning,mansour1995nlog},  communication complexity \cite{grolmusz1997power,ada2012spectral,tsang2013fourier,shpilka2014structure}, and circuit complexity \cite{brandman1990spectral,Bruck90harmonicanalysis,bruck1992polynomial}. In the study of pseudorandomness, it is well-known that small-bias generators\footnote{A small-bias generator with bias $\mu$ outputs a random variable $X\in \zo^n$ such that $\left|\ex{X}{(-1)^{s\cdot X}}\right|\leq \mu$ for every $s\subset [n]$ with $s \ne \emptyset$.} with bias $\eps/L$ (which can be sampled using a seed of length $O(\log (n \cdot L/\eps))$ \cite{NaorNa93,AGHP}) will $\eps$-fool any function whose Fourier mass is at most $L$. 
Width-2 branching programs have Fourier mass at most $O(n)$ 
\cite{BogdanovDvVeYe09,SZ} and are thus fooled by small-bias generators with bias $\eps/n$.  Unfortunately, such a bound does not hold even for very simple width-3 programs. For example, the `mod 3 function,' which indicates when the hamming weight of its input is a multiple of 3 has Fourier mass exponential in $n$.  

However, a more refined measure of Fourier mass is possible and often useful: Let $L^k(f) = \sum_{|s|=k} |\hat{f}[s]|$ be the level-$k$ Fourier mass of $f$.  A bound on the Fourier growth of $f$, or the rate at which $L^k(f)$ grows with $k$, was used by Mansour \cite{mansour1995nlog} to obtain an improved query algorithm for polynomial-size DNF; the junta approximation results of Friedgut \cite{friedgut1998boolean} and Bourgain \cite{bourgain2002distribution} are proven using approximating functions that have slow Fourier growth.  
This notion turns out to be useful in the analysis of pseudorandom generators as well:  Reingold et al.~\cite{ReingoldSteinkeVadhan2013} show that the generator of Gopalan et al.~\cite{GMRTV} will work if there is a good bound on the Fourier mass of low-order coefficients. 
More precisely, they show that for any class $\mathcal{C}$ of functions computed by branching programs that is closed under restrictions and decompositions and satisfies $L^k(f) \leq \poly(n) \cdot c^k$ for every $k$ and $f \in \mathcal{C}$, 
there is a pseudorandom generator with seed length $\tilde{O}(c\cdot \log^2 n)$ that fools every $f\in \mathcal{C}$. 
  They then bound the Fourier growth of permutation branching programs (and the even more general model of ``regular'' branching programs, where each layer is a regular bipartite graph) to obtain a pseudorandom generator for permutation branching programs: 
 
\begin{theorem}[{\cite[Theorem 1.4]{ReingoldSteinkeVadhan2013}}]\label{thm:rsv}
Let $f :\{0,1\}^n \to \{0,1\}$ be computed by a length-$n$, width-$w$, read-once, oblivious, \emph{regular} branching program.  Then, for all $k \in [n]$, $L^k(f) \leq (2w^2)^k$.
\end{theorem}
In particular, the mod 3 function over $O(k)$ bits, which is computed by a permutation branching program of width 3, has Fourier mass $2^{\Theta(k)}$ a level $k$. 
However, the Tribes function,\footnote{The Tribes function (introduced by Ben-Or and Linial \cite{tribes}) is DNF formula where all the terms are the same size and every input appears exactly once. The size of the clauses in this case is chosen to give an asymptotically constant acceptance probability on uniform input.} which is also computed by a width-3 branching program, has Fourier mass $\Theta_k(\log^k n)$ at level $k$, 
so the bound in Theorem \ref{thm:rsv} does not hold for non-regular branching programs even of width 3. 

The Coin Theorem of Brody and Verbin \cite{BrodyVerbin} implies a related result: essentially, a function computed by a width-$w$, length-$n$ branching program cannot distinguish product distributions on $\{0,1\}^n$ any better than a function satisfying $L^k(f) \leq (\log n)^{O(wk)}$ for all $k$. To be more precise, if $X \in \{0,1\}^n$ is $n$ independent samples from a coin with bias $\beta$ (that is, each bit has expectation $(1+\beta)/2$), then $\ex{X}{f[X]} = \sum_s \widehat{f}[s] \beta^{|s|}$. If $L^k(f) \leq (\log n)^{O(wk)}$ for all $k$, then $$\left|\ex{X}{f[X]}-\ex{U}{f[U]}\right| = \left| \sum_{s \ne 0} \widehat{f}[s] \beta^{|s|} \right| \leq \sum_{k \in [n]} L^k(f) |\beta|^{|s|} \leq O(|\beta| (\log n)^{O(w)}),$$ assuming $|\beta| \leq 1/(\log n)^{O(w)}$. Brody and Verbin prove that, if $f$ is computed by a length-$n$, width-$w$ branching program, then $|\ex{X}{f[X]}-\ex{U}{f[U]}| \leq O(|\beta| (\log n)^{O(w)})$.
Since distinguishing product distributions captures much of the power of a class of functions, this leads to the following conjecture. 

\begin{conjecture}[\cite{ReingoldSteinkeVadhan2013}]\label{conj:fouriergrowth}
For every constant $w$, the following holds. Let $f : \{0,1\}^n \to \{0,1\}$ be computed by a width-$w$, read-once, oblivious branching program. Then $$L^k(f) \leq n^{O(1)}\cdot (\log n )^{O(k)},$$ where the constants in the $O(\cdot)$ may depend on $w$.
\end{conjecture}
In this work, we prove this conjecture for $w=3$: 

\begin{theorem}[Fourier Growth of Width 3] \label{MainThmIntro}
Let $f : \{0,1\}^n \to \{0,1\}$ be computed by a width-3, read-once, oblivious branching program. Then, for all $k \in [n]$, $$L^k(f) := \sum_{s : |s|=k} |\widehat{f}[s]| \leq  n^2 \cdot (O(\log n))^k .$$
\end{theorem}
This bound is the main contribution of our work and, when combined with the techniques of Reingold et al.~\cite{ReingoldSteinkeVadhan2013}, implies our main result (Theorem \ref{thm:main-intro}). 

The Tribes function of \cite{tribes} shows that the base of $O(\log n)$ of the Fourier growth in Theorem \ref{MainThmIntro} is tight up to a factor of $\log \log n$. (See Appendix \ref{AppendixOptimal}.)

We also prove Conjecture \ref{conj:fouriergrowth} with $k=1$ for any constant width $w$:
\begin{theorem} \label{TheoremFirstOrder}
Let $f : \{0,1\}^n \to \{0,1\}$ be computed by a width-$w$, length-$n$, read-once, oblivious branching program. Then $$L^1(f) = \sum_{i \in [n]} |\widehat{f}[\{i\}]| \leq O(\log n)^{w-2}.$$
\end{theorem}
The proof is left to Appendix \ref{AppendixFirstOrder}.

\subsection{Techniques}
The intuition behind our approach begins with two extreme cases of width-3 branching programs: permutation branching programs and branching programs in which every layer is a non-permutation layer.     
Permutation branching programs ``mix'' well: on a uniform random input, the distribution over states gets closer to uniform (in $\ell_2$ distance) in each layer. 
We can use this fact with an inductive argument to achieve a bound of $2^{O(k)}$ on the level-$k$ Fourier mass (this is the bound of Theorem \ref{thm:rsv}). 

For branching programs in which \textit{every} layer is a non-permution layer, we can make use of an argument from the work of Brody and Verbin \cite{BrodyVerbin}: when we apply a random restriction (where each variable is kept free with probability roughly $1/k \log n$) to such a branching program, 
the resulting program is `simple' in that the width has collapsed to 2 in many of the remaining layers. This allows us to use arguments tailored to width-2 branching programs, which are well-understood. In particular, we can use the same concept of mixing as used for permutation branching programs. 

To handle general width-3 branching programs, which may contain an arbitrary mix of permutation and non-permutation layers, we group the layers into ``chunks'' containing exactly one non-permutation layer each.  Instead of using an ordinary random restriction, we consider a series of restrictions similar to those in Steinberger's ``interwoven hybrids'' technique \cite{Steinberger} (in our argument each chunk will correspond to a single layer in \cite{Steinberger}).  In Section \ref{subsect:width-reduction}, we use such restrictions to show that the level-$k$ Fourier mass of an arbitrary width-3 program can be bounded in terms of the level-$k$ Fourier mass of a program $D$ which has the following ``pseudomixing'' form: $D$ can be split into $r\in [n]$ branching programs $D_1\circ D_2 \circ \cdots \circ D_r$, where each $D_i$ has at most $3k$ non-regular layers and the layer splitting consecutive $D_i$s has width 2.

We then generalize the arguments used for width-2 branching programs to ``pseudomixing'' branching programs. We can show that each chunk $D_i$ is either mixing or has small Fourier growth, which suffices to bound the Fourier growth of $D$.
\subsection{Organization}

In Section \ref{SectionTechniques} we introduce the definitions and tools we use in our proof. In Section \ref{SubSectionBP} we formally define branching programs and explain our view of them as matrix-valued functions. In Sections \ref{SubSectionFourier} and \ref{SectionFourierBounds} we define the matrix-valued Fourier transform and explain how we use it.


We prove the upper bound on Fourier mass of oblivious, read-once, width-3 branching programs 
(i.e., Theorem \ref{MainThmIntro}) in Section \ref{SubSectionLowOrder} (Theorem \ref{TheoremLow}). 
In Sections \ref{SubSectionPRrestriction} 
and Section \ref{SubSectionPRG} we construct and analyse our pseudorandom generator, which proves the main result (Theorem \ref{thm:main-intro}).
The proof of Theorem \ref{TheoremFirstOrder} is left to Appendix \ref{AppendixFirstOrder}.


\section{Preliminaries} \label{SectionTechniques}

\subsection{Branching Programs} \label{SubSectionBP}

We define a length-$n$, width-$w$ \textbf{program} to be a function $B : \{0,1\}^n \times [w] \to [w]$, which takes a start state $u \in [w]$ and an input string $x \in \{0,1\}^n$ and outputs a final state $B[x](u)$.

Often we think of $B$ as having a fixed \textbf{start state} $u_0$ and a set $S \subset [w]$ of \textbf{accept states}. Then $B$ \textbf{accepts} $x \in \{0,1\}^n$ if $B[x](u_0) \in S$. We say that $B$ \textbf{computes the function} $f : \{0,1\}^n \to \{0,1\}$ if $f(x)=1$ if and only if $B[x](u_0) \in S$.

In our applications, the input $x$ is randomly (or pseudorandomly) chosen, in which case a program can be viewed as a Markov chain randomly taking initial states to final states. For each $x \in \{0,1\}^n$, we let $B[x] \in \{0,1\}^{w \times w}$ be a matrix defined by $$B[x](u,v) = 1 \iff B[x](u)=v.$$ 

For a random variable $X$ on $\{0,1\}^n$, we have $\ex{X}{B[X]} \in [0,1]^{w \times w},$ where $\ex{R}{f(R)}$ is the {expectation} of a function $f$ with respect to a random variable $R$. Then the entry in the $u^\text{th}$ row and $v^\text{th}$ column $\ex{X}{B[X]}(u,v)$ is the probability that $B$ takes the initial state $u$ to the final state $v$ when given a random input from the distribution $X$---that is, $$\ex{X}{B[X]}(u,v) = \pr{X}{B[X](u)=v},$$ where $\pr{R}{e(R)}$ is the {probability} of an event $e$ with respect to the random variable $R$.

A branching program reads one bit of the input at a time (rather than reading $x$ all at once) maintaining only a state in $[w] = \{1,2, \cdots, w\}$ at each step. We capture this restriction by demanding that the program be composed of several smaller programs, as follows.

Let $B$ and $B'$ be width-$w$ programs of length $n$ and $n'$ respectively. We define the \textbf{concatenation} $B \circ B' : \{0,1\}^{n+n'} \times [w] \to [w]$ of $B$ and $B'$ by $$(B \circ B')[x \circ x'](u) := B'[x'](B[x](u)),$$ which is a width-$w$, length-$(n+n')$ program. That is, we run $B$ and $B'$ on separate inputs, but the final state of $B$ becomes the start state of $B'$. Concatenation corresponds to matrix multiplication---that is, $(B \circ B')[x \circ x'] = B[x] \cdot B'[x']$, where the two programs are concatenated on the left hand side and the two matrices are multiplied on the right hand side. 

A length-$n$, width-$w$, \textbf{ordered branching program} (abbreviated \textbf{OBP}) is a program $B$ that can be written $B = B_1 \circ B_2 \circ \cdots \circ B_n$, where each $B_i$ is a length-$1$ width-$w$ program. We refer to $B_i$ as the $i^\text{th}$ \textbf{layer} of $B$. We denote the \textbf{subprogram} of $B$ from layer $i$ to layer $j$ by $B_{ i \cdots j} := B_i \circ B_{i+1} \circ \cdots \circ B_j$.

A length-$n$, width-$w$, ordered branching program can also be viewed as a directed acyclic graph. The vertices are arranged into $n+1$ layers each of size $w$. The edges go from one layer to the next. In particular, there is an edge labelled $b$ from vertex $u$ in layer $i-1$ to vertex $B[b](u)$ in layer $i$. 

\begin{figure}[h]
\centering
\includegraphics[scale=.25]{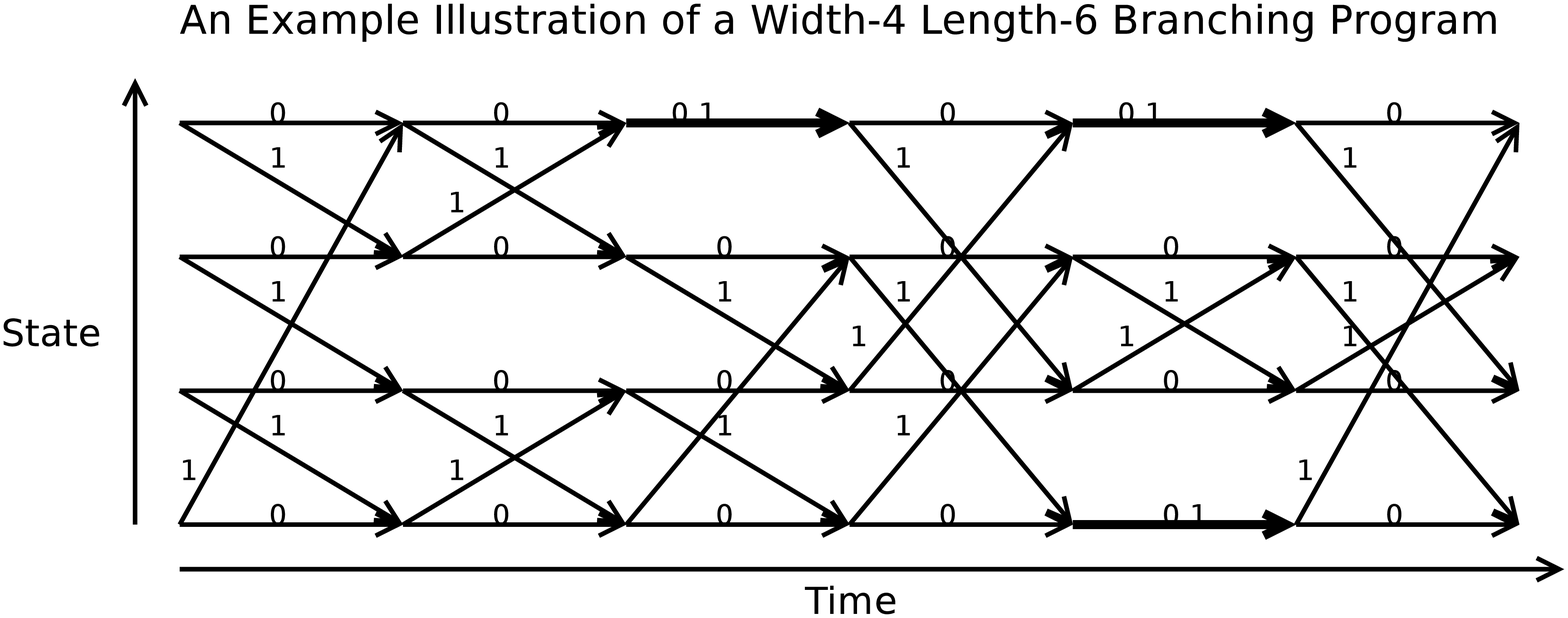}
\end{figure}

We use the following notational conventions when referring to layers of a length-$n$ branching program. We need to distinguish between layers of vertices and layers of edges (although this is often clear from context). Layers of \emph{edges} are the $B_i$s and are numbered from $1$ to $n$. Layers of \emph{vertices} are the states between the $B_i$s and are numbered from $0$ to $n$. The edges in layer $i$ ($B_i$) go from vertices in layer $i-1$ to vertices in layer $i$. 

General read-once, oblivious branching programs (a.k.a. unordered branching programs) can be reduced to the ordered case by a permutation of the input bits. Formally, a \textbf{read-once, oblivious branching program} $B$ is an ordered branching program $B'$ composed with a permutation $\pi$. That is, $B[x]=B'[\pi(x)]$, where the $i^\text{th}$ bit of $\pi(x)$ is the $\pi(i)^\text{th}$ bit of $x$

For a program $B$ and an arbitrary distribution $X$, the matrix $\ex{X}{B[X]}$ is \textbf{stochastic}---that is, $$\sum_v \ex{X}{B[X]}(u,v) = 1$$ for all $u$ and $\ex{X}{B[X]}(u,v) \geq 0$ for all $u$ and $v$. A program $B$ is called a \textbf{regular program} if the matrix $\ex{U}{B[U]}$ is \textbf{doubly stochastic}---that is, both $\ex{U}{B[U]}$ and its transpose $\ex{U}{B[U]}^*$ are stochastic. A program $B$ is called a \textbf{permutation program} if $B[x]$ is a permutation matrix for every $x$ or, equivalently, $B[x]$ is doubly stochastic. Note that a permutation program is necessarily a regular program and, if both $B$ and $B'$ are regular or permutation programs, then so is their concatenation.

A regular program $B$ has the property that the uniform distribution is a stationary distribution of the Markov chain $\ex{U}{B[U]}$, whereas, if $B$ is a permutation program, the uniform distribution is stationary for $\ex{X}{B[X]}$ for \emph{any} distribution $X$.

A \textbf{regular branching program} is a branching program where each layer $B_i$ is a regular program and likewise for a \textbf{permutation branching program}. We will refer to layer $i$ as regular if $B_i$ is a regular program and we say that layer $i$ is non-regular otherwise.

Equivalently, a regular branching program is one where, in the directed acyclic graph, each vertex has in-degree 2 (in addition to having out-degree 2) except those in the start layer -- that is, each layer of edges is a regular graph (hence the name). A permutation branching program has the additional constraint that the incoming edges have distinct labels.

We also consider branching programs of varying width -- some layers have more vertices than others. The overall width of the program is the maximum width of any layer.  This means that the edge layers $B_i$ may give non-square matrices. For $i \in [n]$, if $B_i[x] \in \{0,1\}^{w \times w'}$, then we refer to $w$ as the width of layer $i-1$ and $w'$ as the width of layer $i$.

\subsection{Norms}

We are interested in constructing a random variable $X$ (the output of the pseudorandom generator) such that $\ex{X}{B[X]} \approx \ex{U}{B[U]}$, where $U$ is uniform on $\{0,1\}^n$. Throughout we use $U$ to denote the \textbf{uniform distribution}. The error of the pseudorandom generator will be measured by a norm of the matrix $\ex{X}{B[X]} - \ex{U}{B[U]}$.

For a matrix $A \in \mathbb{R}^{w \times w'}$, define the \textbf{$\rho$ operator norm} of $A$ by $$\norm{A}_\rho = \max_x \frac{\norm{xA}_\rho}{\norm{x}_\rho},$$ where $\rho$ specifies a vector norm (usually $1$, $2$, or $\infty$ norm). Define the \textbf{Frobenius norm} of $A \in \mathbb{R}^{w \times w'}$ by $$\frob{A}^2 = \sum_{u,v} A(u,v)^2 = \text{trace}(A^*A) = \sum_\lambda |\lambda|^2,$$ where $A^*$ is the (conjugate) transpose of $A$ and the last sum is over the singular values $\lambda$ of $A$. Note that $\norm{A}_2 \leq \frob{A}$ for all $A$.


\subsection{Fourier Analysis} \label{SubSectionFourier}

Let $B : \{0,1\}^n \to \mathbb{R}^{w\times w'}$ be a matrix-valued function (such as given by a length-$n$, width-$w$ branching program). Then we define the \textbf{Fourier transform} of $B$ as a matrix-valued function $\widehat{B} : \{0,1\}^n \to \mathbb{R}^{w \times w'}$ given by $$\widehat{B}[s] := \ex{U}{B[U] \chi_s(U)},$$ where $s \in \{0,1\}^n$ (or, equivalently, $s \subset [n]$) and $$\chi_s(x) = (-1)^{\sum_i x(i) \cdot s(i)} = \prod_{i \in s} (-1)^{x(i)}.$$ We refer to $\widehat{B}[s]$ as the $s^\text{th}$ \textbf{Fourier coefficient} of $B$. The \textbf{order} of a Fourier coefficient $\widehat{B}[s]$ is $|s|$---the \textbf{Hamming weight} of $s$, which is the size of the set $s$ or the number of $1$s in the string $s$. Note that this is equivalent to taking the real-valued Fourier transform of each of the $w \cdot w'$ entries of $B[x]$ separately, but we will see below that this matrix-valued Fourier transform is nicely compatible with matrix algebra.

For a random variable $X$ over $\{0,1\}^n$ we define its $s^\text{th}$ \textbf{Fourier coefficient} as $$\widehat{X}(s) := \ex{X}{\chi_s(X)},$$ which, up to scaling, is the same as taking the real-valued Fourier transform of the probability mass function of $X$.
We have the following useful properties.
\begin{lemma} \label{LemmaFourier}
Let $A: \{0,1\}^n \to \mathbb{R}^{w\times w'}$ and $B: \{0,1\}^n \to \mathbb{R}^{w'\times w''}$ be matrix valued functions. Let $X$, $Y$, and $U$ be independent random variables over $\{0,1\}^n$, where $U$ is uniform. Let $s,t \in \{0,1\}^n$. Then we have the following.
\begin{itemize}
\item Decomposition: If $C[x \circ y] = A[x] \cdot B[y]$ for all $x,y \in \{0,1\}^n$, then $\widehat{C}[s \circ t] = \widehat{A}[s] \cdot \widehat{B}[t]$.
\item Expectation: $\ex{X}{B[X]} = \sum_s \widehat{B}[s] \widehat{X}(s)$.
\item Fourier Inversion for Matrices: $B[x] = \sum_s \widehat{B}[s] \chi_s(x)$.
\item Fourier Inversion for Distributions: $\pr{X}{X=x} = \ex{U}{\widehat{X}(U) \chi_U(x)}$.
\item Convolution for Distributions: If $Z = X \oplus Y$, then $\widehat{Z}(s) = \widehat{X}(s) \cdot \widehat{Y}(s)$.
\item Parseval's Identity: $\sum_{s \in \{0,1\}^n} \frob{\widehat{B}[s]}^2 = \ex{U}{\frob{B[U]}^2}$.
\item Convolution for Matrices: If, for all $x \in \{0,1\}^n$, $C[x] = \ex{U}{A[U] \cdot B[U \oplus x]}$, then $\widehat{C}[s] = \widehat{A}[s] \cdot \widehat{B}[s]$.
\end{itemize}
\end{lemma}

The Decomposition property is what makes the matrix-valued Fourier transform more convenient than separately taking the Fourier transform of the matrix entries as done by Bogdanov et al.~\cite{BPW}. If $B$ is a length-$n$ width-$w$ branching program, then, for all $s \in \{0,1\}^n$, $$\widehat{B}[s] = \widehat{B}_1[s_1] \cdot \widehat{B}_2[s_2] \cdot \cdots \cdot \widehat{B}_n[s_n].$$

\subsection{Small-Bias Distributions} \label{SubSectionSmallBias}

The \textbf{bias} of a random variable $X$ over $\{0,1\}^n$ is defined as $$\text{bias}(X) := \max_{s \ne 0} |\widehat{X}(s)|.$$ A distribution is \textbf{$\varepsilon$-biased} if it has bias at most $\varepsilon$. Note that a distribution has bias $0$ if and only if it is uniform. Thus a distribution with small bias is an approximation to the uniform distribution. We can sample an $\varepsilon$-biased distribution $X$ on $\{0,1\}^n$ with seed length $O(\log(n/\varepsilon))$ and using space $O(\log(n/\varepsilon))$ \cite{NaorNa93,AGHP}.

Small-bias distributions are useful pseudorandom generators: A $\varepsilon$-biased random variable $X$ is indistinguishable from uniform by any linear function (a parity of a subset of the bits of $X$). That is, for any $s \subset [n]$, we have $\left| \ex{X}{\bigoplus_{i \in s} X_i} - 1/2\right|\leq 2\varepsilon$. Small bias distributions are known to be good pseudorandom generators for width-$2$ branching programs \cite{BogdanovDvVeYe09}, but not width-$3$. For example, the uniform distribution over $\{x \in \{0,1\}^n : |x| \mod{3} = 0 \}$ has bias $2^{-\Theta(n)}$, but does not fool width-$3$, ordered, permutation branching programs.

\subsection{Fourier Mass} \label{SectionFourierBounds}

We analyse small bias distributions as pseudorandom generators for branching programs using Fourier analysis. Intuitively, the Fourier transform of a branching program expresses that program as a linear combination of linear functions (parities), which can then be fooled using a small-bias space.

Define the \textbf{Fourier mass} of a matrix-valued function $B$ to be $$L(B) := \sum_{s \ne 0} \norm{\widehat{B}[s]}_2.$$ Also, define the \textbf{Fourier mass of $B$ at level $k$} as $$L^k(B) := \sum_{s \in \{0,1\}^n : |s|=k} \norm{\widehat{B}[s]}_2.$$ Note that $L(B) =\sum_{k \geq 1} L^k(B)$. We define $L^{\geq k}(B) := \sum_{k' \geq k} L^{k'}(B)$ and $L^{\leq k}(B)$, $L^{>k}(B)$, $L^{<k}(B)$ are defined analogously.

The Fourier mass is unaffected by order:

\begin{lemma} \label{LemmaFourierPermutation}
Let $B, B' : \{0,1\}^n \to \mathbb{R}^{w \times w}$ be matrix-valued functions satisfying $B[x]=B'[\pi(x)]$, where $\pi : [n] \to [n]$ is a permutation. Then, for all $s \in \{0,1\}^n$, $\widehat{B}[s]=\widehat{B'}[\pi(s)]$. In particular, $L(B)=L(B')$ and $L^k(B)=L^k(B')$ for all $k$ and $\rho$.
\end{lemma}

Lemma \ref{LemmaFourierPermutation} implies that the Fourier mass of any read-once, oblivious branching program is equal to the Fourier mass of the corresponding ordered branching program.

If $L(B)$ is small, then $B$ is fooled by a small-bias distribution:

\begin{lemma} \label{LemmaBiasMass}
Let $B$ be a length-$n$, width-$w$ branching program. Let $X$ be a $\varepsilon$-biased random variable on $\{0,1\}^n$. We have $$\norm{\ex{X}{B[X]}-\ex{U}{B[U]}}_2 = \norm{\sum_{s \ne 0} \widehat{B}[s] \widehat{X}(s)}_2 \leq L(B) \varepsilon.$$
\end{lemma}

In the worst case $L(B) = 2^{\Theta(n)}$, even for a length-$n$ width-$3$ permutation branching program $B$. For example, the program $B_{\text{mod 3}}$ that computes the Hamming weight of its input modulo $3$ has exponential Fourier mass.

We show that, using `restrictions,' we can ensure that $L(B)$ is small.

\section{Fourier Analysis of Width-3 Branching Programs} \label{SubSectionLowOrder}

In this section we prove a bound on the low-order Fourier mass of width-3, read-once, oblivious branching programs. This is key to the analysis of our pseudorandom generator. Improvements to this result directly translate to improvements in our final result.

\begin{theorem} \label{TheoremLow}
Let $f : \{0,1\}^n \to \{0,1\}$ be computed by a width-3, read-once, oblivious branching program. Then, for all $k \in [n]$, $$L^k(f) \leq 8 n^2 \cdot \left( C \cdot \log_2(3 n) \right)^k = n^2 \cdot \left( O(\log n) \right)^k,$$ where $C$ is a universal constant.\footnote{We have not optimised any constants and only show $C \leq 10^7$.}
\end{theorem}

To prove Theorem \ref{TheoremLow} we will consider the matrix valued function $B$ of the branching program computing $f$.  Note that $|\hat{f}[s]| \leq ||\widehat{B}[s]||_2$ for all $s$ so $L^k(f) \leq L^k(B)$.  We may also assume without loss of generality that the first and last layers of the program have width 2 (there is only one start state, and there are at most 2 accept states otherwise the program is trivial).          
The proof proceeds in two parts. The first part reduces the problem to one about branching programs of a special form, namely ones where many layers have been reduced to width-2. The second part uses the mixing properties of width-2 programs to bound the Fourier mass. 

\subsection{Part 1 -- Reduction of Width by Random Restriction} \label{subsect:width-reduction}

Our reduction can be stated as follows.
\begin{proposition} \label{PropositionPart1}
Let $B$ be a length-$n$ width-3 ordered branching program (abbreviated \textbf{3OBP}), $m \geq k$, and $k \in [n]$ with the first and last layers having width at most 2. 
Then
$$L^k(B) \leq n\cdot \binom{m}{k} \sum_{\ell \geq 0}  2^{-\ell(m-k)} L^k(D^{6(\ell+1)k})$$ 
where each $D^{6(\ell+1)k} = D_1^{6(\ell+1)k} \circ D_2^{6(\ell+1)k} \circ \cdots \circ D_r^{6(\ell+1)k}$, where $r \in [n]$, each $D_i^{6(\ell+1)k}$ is a 3OBP with at most $6(\ell+1) k$ non-regular layers, and the first and last layers of each $D_i$ have width at most 2.
\end{proposition}

In Section \ref{SubSectionMixing}, we will prove $L^k(D^{6(\ell+1)k}) \leq n \cdot O(\ell)^k$. Taking $m=2k$, this implies $L^k(B) \leq n^2 \cdot O(k)^k$. Finally, in Section \ref{SubSectionBootstrapping}, we show that we may assume $k \leq O(\log n)$, so we get a Fourier growth bound of $L^k(B) \leq n^2 \cdot O(\log n)^k$.
Here we focus on the proof of Proposition \ref{PropositionPart1}.

First some definitions:



For $g \subset [n]$
--  and $x \in \{0,1\}^n$, define the \textbf{restriction of $B$ to $g$ using $x$ -- denoted $B|_{\overline{g} \leftarrow x}$ --} to be the branching program obtained by setting the inputs (layers of edges) of $B$ outside $g$ to values from $x$ and leaving the inputs in $g$ free. More formally, $$B|_{\overline{g} \leftarrow x}[y] = B[\Set(g,y,x)],$$ where $$\Set(g,y,x)_i = \left\{ \begin{array}{cl} y_i & i \in g \\ x_i & i \notin g \\ \end{array} \right\}.$$

We prove Proposition \ref{PropositionPart1} by considering a restriction $B|_{\overline{g} \leftarrow x}$ for a carefully chosen $g$ and a random $x$. We show (Lemma \ref{LemmaInterwoven}) that is suffices to bound the Fourier growth of the restricted program $B|_{\overline{g} \leftarrow x}$ and (Lemma \ref{LemmaRestrict}) that the restricted program is of the desired form $D^{6(\ell+1)k}$ with high probability. 

Define a \textbf{chunk} to be a 3OBP with exactly one non-regular layer. An \textbf{$l$-chunk} 3OBP is a 3OBP $B$ such that $B = C_1 \circ C_2 \circ \cdots \circ C_l$, where each $C_i$ is a chunk. Equivalently, an $l$-chunk 3OBP is a 3OBP with exactly $l$ non-regular layers. The partitioning of $B$ into chunks is not necessarily unique. But we fix one such partitioning for each 3OBP and simply refer to the $i^\text{th}$ chunk $C_i$. If $B$ is an $l$-chunk length-$n$ 3OBP, let $c_i \subset [n]$ be the coordinates corresponding to $C_i$.

We will compute a bound on the level-$k$ Fourier weight of $B$ via a series of ``interwoven'' restrictions similar to Steinberger's technique \cite{Steinberger}. Lemma \ref{LemmaInterwoven} below tells us that we may obtain a bound by bounding, in expectation, the level-$k$ weight of a restricted branching program.  We then argue that with high probability over this restriction, the width of the resulting program will be essentially reduced.  In particular, there is a layer of width 2 after every $O(k)$ non-regular layers.


We now describe some notation that will be used for the interwoven restrictions.  
For $t \subset [m]$, define $$g_t := \bigcup_{(i \text{ mod } m)+1 \in t} c_i$$ and $$G_t^k := \{s \subset g_t : |s| = k\}.$$ We refer to $g_t$ as the \textbf{$t^\text{th}$ group of indices} and $G_t^k$ as the \textbf{$t^\text{th}$ group of (order $k$) Fourier coefficients}.
The following lemma tells us that we may bound the level-$k$ Fourier weight by considering a fixed subset $t\subset [m]$ of size $k$ and the level-$k$ Fourier weight of the branching program that results by randomly restricting the variables outside of $g_t$:       

\begin{lemma}\label{LemmaInterwoven} Let $B$ be a length-$n$ 3OBP, $k\in [n]$, $m\geq k$ and $g_t$ as above.  Then  
$$ L^k(B)\leq \binom{m}{k} \cdot ~\max_{t \subset [m] : |t| = k}~ \ex{U}{L^k(B|_{\overline{g_t} \leftarrow U})}.$$
\end{lemma}

\begin{proof}
Note that some Fourier coefficients appear in multiple $G_t$s, but every coefficient at level $k$ appears in at least one $G_t$.  Thus 
\begin{align*}
L^k(B) & \leq \sum_{t : |t| = k} \sum_{s \in G_t^k} \norm{\widehat{B}[s]}_2 \\
& = \sum_{t : |t| = k} \sum_{s \in G_t^k} \norm{\ex{U}{\widehat{B_{\overline{g_t} \leftarrow U}}[s]}}_2 \\ 
& \leq \sum_{t : |t| = k} \sum_{s \in G_t^k} \ex{U}{\norm{\widehat{B_{\overline{g_t} \leftarrow U}}[s]}_2} \\ 
& = \sum_{t : |t| = k} \ex{U}{L^k(B|_{\overline{g_t} \leftarrow U})} \\ 
& \leq \binom{m}{k} \max_{t \subset [m] : |t| = k} \ex{U}{L^k(B|_{\overline{g_t} \leftarrow U})},
\end{align*}
where the second inequality follows from the convexity of the norm.
\end{proof}

Given Lemma \ref{LemmaInterwoven}, we may now 
prove Proposition  \ref{PropositionPart1} by giving an upper bound on 
$\ex{U}{L^k(B|_{\overline{g_t} \leftarrow U})}$ for any fixed $t \subset [m]$ with $|t|=k$.
To do this, we prove that a random restriction to $\overline{g_t}$ will, with high probability, result in a branching program of the desired form. 
\begin{lemma}\label{LemmaRestrict}
Let $B$ be a length-$n$ 3OBP, $k,\ell\in [n]$, $m\geq k$ and fix $t\subseteq [m]$ with $|t|=k$.  
Then with probability at least $1-n\cdot 2^{-\ell \cdot (m-k)}$ over a random choice of $x \in \{0,1\}^n$, $$ B|_{\overline{g_t}\leftarrow x} = D_1 \circ D_2\circ \cdots \circ D_r,$$
where $r\in [n]$ and each $D_i$ is a 3OBP with at most $6\ell k$ non-regular layers and the layer of vertices between $D_{i-1}$ and $D_i$ have width at most $2$. 
\end{lemma}
\begin{proof}
Let $t = \{t_1, t_2, \cdots, t_k\}$ where $t_1 < t_2 < \cdots < t_k$. Define $t_{a+k} = t_a + m$ for all $a$. Let $a'$ be the largest value of $a$ such that $t_a \leq n$. We redefine $t_{a'+1}=n+1$. We can write $$B|_{\overline{g_t} \leftarrow x} = C_{t_1}' \circ C_{t_2}' \circ \cdots \circ C_{t_{a'}}',$$ where each $C_{t_a}'$ corresponds to $C_{t_a}$. However, the chunks $C_{t_a+1}, C_{t_a+2}, \cdots C_{t_{a+1}-1}$ have been restricted and $C_{t_a}'$ reflects that. Formally, for all $a \in [l'] \backslash \{1\}$ and  $y \in \{0,1\}^{|c_{t_a}|}$, we have $$C_{t_a}'[y] = C_{t_a}[y] \cdot C_{t_{a}+1}[x_{c_{t_{a}+1}}] \cdot \cdots \cdot C_{t_{a+1}-1}[x_{c_{t_{a+1}-1}}],$$ and, for all $y \in \{0,1\}^{|c_{t_{1}}|}$, we have $$C_{t_{1}}'[y] = C_{1}[x_{c_{1}}] \cdot \cdots \cdot C_{t_{1}-1}[x_{c_{t_{1}-1}}] \cdot C_{t_{1}}[y] \cdot C_{t_{{1}}+1}[x_{c_{t_{1}+1}}] \cdot \cdots \cdot C_{t_2-1}[x_{c_{t_2-1}}],$$
where $x_{c_t}$ is the coordinates of $x$ contained in $c_t$. Moreover, we remove any unreachable vertices. That is, if $\ex{U}{C'_{t_a}[U]}$ has a column of zeros, we remove that column (making the matrix non-square) and remove the corresponding row from $C'_{t_{a+1}}$. This reduces the width of the layer of vertices between $C'_{t_a}$ and $C'_{t_{a+1}}$. 

It should be clear that each $C_{t_a}'$ is a 3OBP with between one and three non-regular layers. (One non-regular layer comes from $C_{t_a}$ and the first and last layers may become non-regular after the restriction.)

Consider $\ell\cdot k+1$ consecutive $C_{t_a}'$s in the restricted program $C_{t_a}' \circ C_{t_{a+1}}' \circ \cdots \circ C_{t_{a+\ell\cdot k}}'$ and the corresponding subprogram before restriction $C_{t_a}\circ\cdots \circ C_{t_{a+\ell\cdot k}}$ which contains $\ell\cdot m+1$ chunks (recall that $t_{a+k}=t_a+m$).  There are exactly $\ell\cdot k+1$ chunks $C_{t_a}',\dots,C_{t_a+k}',\dots, C_{t_{a+\ell\cdot k}}'$ which remain free after the restriction, i.e., $\ell (m-k)$ chunks have been restricted.  Furthermore, each such chunk contains a non-regular layer.  

Each non-regular layer that is restricted has at least a $1/2$ probability of reducing the width: Being non-regular implies that there are two edges with the same label going to the same vertex. There is a 1/2 probability of the restriction picking that label. If this happens, then one of the vertices on the right side of the layer becomes unreachable and therefore the width is reduced. This is the same argument that is used by Brody and Verbin \cite{BrodyVerbin} and Steinberger \cite{Steinberger}.

So, with probability at least $1-2^{-\ell(m-k)}$ over the choice of $x$, there exists a layer 
between $C_{t_{a}}'$ and $C_{t_{a+\ell\cdot k}}'$ that has width at most 2. Call such a layer of vertices a \textbf{bottleneck}.

By a union bound, with probability at least $1 - n \cdot 2^{-\ell(m-k)}$ over the choice of $x$, every $\ell\cdot k+1$ successive $C_{t_a}'$s have one such bottleneck. We split the $C_{t_a}'$s into groups that are separated by bottlenecks to write  
$B|_{\overline{g_t} \leftarrow x} = D_1 \circ D_2 \circ \cdots \circ D_{r}$ (one $D_i$ per group).  
Since each remaining chunk has at most 3 non-regular layers, and there can be at most $2\ell k$ chunks in each group, each $D_i$ is a 3OBP with at most $6\ell k$ non-regular layers and the layer between $D_i$ and $D_{i+1}$ has width at most 2 (i.e., is a bottleneck).  
\end{proof}
We may now complete the proof of Proposition \ref{PropositionPart1}.
\begin{proof}[Proof of Proposition \ref{PropositionPart1}]
By Lemma \ref{LemmaInterwoven}, it suffices to bound, for every fixed $t$, the quantity
 $\ex{U}{L^k(B|_{\overline{g_t} \leftarrow U})}.$  
%
  
 We compute the expectation of $L^k(B|_{\overline{g_t}\leftarrow U})$ by conditioning on the how far apart bottlenecks are in the restricted program and applying Lemma \ref{LemmaRestrict}.  Let $\beta_x$ be the largest number of non-regular layers occuring in $B|_{\overline{g_t} \leftarrow x}$ that are not separated by a bottleneck (i.e., a width-2 layer).
 \begin{align*}
\ex{U}{L^k(B|_{\overline{g_t} \leftarrow U})}  & \leq \sum_{\ell \geq 0} 
\pr{x}{ 6\ell k < \beta_x \leq  6(\ell+1)k} \cdot \ex{x}{L^k(B|_{\overline{g_t} \leftarrow x}) \mid \beta_x \leq 6(\ell+1)k}\\
& \leq \sum_{\ell \geq 0} n\cdot 2^{-\ell \cdot (m-k)} \cdot L^k(D^{6(\ell+1)k})
\end{align*}
 where $D^{6(\ell+1)k}$ is the branching program that maximizes $L^k(B|_{\overline{g_t} \leftarrow x})$ over $x$ such that $\beta_x \leq 6(\ell+1)k$.
\end{proof}

\subsection{Part 2 -- Mixing in Width-2} \label{SubSectionMixing}

Now it remains to bound the Fourier mass of 3OBPs of the form given by Proposition \ref{PropositionPart1}. 

\begin{proposition} \label{PropositionDmass}
Let $D^\ell$ be a length-$n$ 3OBP such that $$D^\ell = D_1^\ell \circ D_2^\ell \circ \cdots \circ D_r^\ell,$$ where each $D_i^\ell$ is a 3OBP with at most $\ell$ non-regular layers and width 2 in the first and last layers. Then $L^k(D^\ell) \le 2n \cdot (6000(\ell+1))^k = n \cdot O(\ell)^k$ for all $k,\ell \geq 1$.
\end{proposition}

Before we prove Proposition \ref{PropositionDmass}, we show how it implies a bound on Fourier mass of general 3OBPs.

\begin{proposition} \label{PropositionLow}
Let $B$ be a length-$n$, width-3, read-once, oblivious branching program with width 2 in the first and last layers. Then, for all $k \in [n]$,
 $$L^k(B) := \sum_{s \in \{0,1\}^n : |s|=k} |\widehat{B}[s]| \leq 8 n^2 \cdot \left(200000 k \right)^k  = n^2 \cdot O(k)^k.$$
\end{proposition}
\begin{proof}
Let $B$ be a 3OBP computing $f$ and assume $B$ has width 2 in the first and last layers. By Propositions \ref{PropositionPart1} and \ref{PropositionDmass}, we have, setting $m=2k+1$, 
\begin{align*}
L^k(B) \leq& n\cdot \binom{m}{k} \sum_{\ell \geq 0}  2^{-\ell(m-k)} L^k(D^{6(\ell+1)k})\\
\leq& n\cdot \binom{m}{k} \sum_{\ell \geq 0}  2^{-\ell(m-k)} 2n \cdot (6000(6(\ell+1)k+1))^k\\
\leq& 2n^2 \binom{2k+1}{k} \sum_{\ell \geq 0}  2^{-\ell (k+1)} (6000(6(\ell+1)k+1))^k\\
\leq& 4n^2 4^k \sum_{\ell \geq 0}  2^{-\ell} \cdot \left( 6000\frac{6(\ell+1)k+1}{2^\ell} \right)^k\\
\leq& 4n^2 4^k \left(\sum_{\ell \geq 0}  2^{-\ell}\right) \cdot \left( 6000 \cdot 7k \right)^k\\
\leq& 8n^2 \cdot \left( 6000 \cdot 4 \cdot 7k \right)^k,
\end{align*}
as required.
\end{proof}

A key notion in our proof is a measure of the extent to which a branching program (or subprogram) mixes, and the way this is reflected in the Fourier spectrum.    
 For an ordered branching program $D$ of width $w$, define 
 $$\lambda(D) = \max_{x \in \RR^w : \sum_i x_i = 0} \frac{\norm{x\ex{U}{D[U]}}_2}{\norm{x}_2}$$
The quantity $\lambda(D)$ is a measure of the \textbf{mixing} of $D$. If $D$ is regular, we have $\lambda(D) \in [0,1]$, where $0$ corresponds to perfect mixing and $1$ to no mixing. If $D$ is not regular, it is possible that $\lambda(D)>1$. However, for width-$2$ -- where $\ex{U}{D[U]}$ is a $2\times 2$ matrix -- it turns out that $\lambda(D) \leq 1$ even if $D$ is non-regular. In particular, $$\text{if}~~\ex{U}{D[U]} = \left( \begin{array}{cc} 1-\alpha & \alpha \\ \beta & 1- \beta \end{array} \right), ~~\text{then}~~ \lambda(D) = \frac{\norm{ (1, -1) \ex{U}{D[U]} }_2}{\norm{(1,-1)}_2} = |1-\alpha-\beta|.$$ The rows of $\ex{U}{D[U]}$ must sum to $1$ and have non-negative entries (as they are a probability distribution). So $\alpha, \beta \in [0,1]$, which implies $\lambda(D) \leq 1$. This fact is crucial to our analysis and is the main reason our results do not extend to higher widths. 

Note that for any $s \ne 0$, the rows of $\widehat{D}[s]$ sum to zero.
Thus for any branching program $D=D_1\circ D_2$ and coefficient $\widehat{D}[s]$ with $s=(s_1,s_2)$ satisfying $s_2=0$, we have
\begin{equation}\label{eqn:mix-coeff}
\norm{\widehat{D}[s]}_2\leq \norm{\widehat{D_1}[s_1]}_2\cdot \lambda(D_2).
\end{equation}
For branching programs $B$ in which every layer is mixing -- that is $\lambda(B_i) \leq C < 1$ for all $i$ -- this fact can be used with an inductive argument (simpler than the proof below) to obtain a $1/(1-C)^{O(k)}$ bound on the level-$k$ Fourier mass.  We show that any $D_i$ in the branching program of the form given by Proposition \ref{PropositionPart1} will either mix well or have small Fourier mass after restriction.       
More precisely, 
 define the \textbf{$p$-damped Fourier mass} of a branching program $B$ as $$L_p(B) = \sum_{k > 0} p^k L^k(B) = \sum_{s \ne 0} p^{|s|} \norm{\widehat{B}[s]}_2.$$ Note that $L^k(B) \leq L_p(B) p^{-k}$ for all $k$ and $p$.
The main lemma we prove in this section is the following.
\begin{lemma} \label{LemmaLL}
If $D$ is a length-$d$ 3OBP with $k \geq 1$ non-regular layers that has only two vertices in the first and last layers, then $$\lambda(D) + L_p(D) \leq 1$$ for any $p \leq 1/6000(k+1)$. 
\end{lemma}

First, we show that Lemma \ref{LemmaLL} implies Proposition \ref{PropositionDmass}:
\begin{proof}[Proof of Proposition \ref{PropositionDmass}]
We inductively show that $$L_p(D_1^\ell \circ \cdots\circ D_i^\ell) \leq 2i,$$
and hence $L_p(D) \leq 2r \leq 2n$. 
For $i=0$ this is trivial. Now suppose it holds for $i$. By decomposition (Lemma \ref{LemmaFourier}), we have
\begin{align*}
L_p(D_1^\ell \cdots D_i^\ell \circ D_{i+1}^\ell) =& \sum_{(s,t) \ne 0} p^{|s|+|t|} \norm{\widehat{D_1^\ell \cdots D_i^\ell}[s] \cdot \widehat{D_{i+1}^\ell}[t]}_2\\
\leq& \sum_{s \ne 0} p^{|s|} \norm{\widehat{D_1^\ell \cdots D_i^\ell}[s] }_2 \sum_{t \ne 0} p^{|t|} \norm{\widehat{D_{i+1}^\ell}[t]}_2 \\&+  \sum_{s \ne 0} p^{|s|} \norm{\widehat{D_1^\ell \cdots D_i^\ell}[s] \cdot\widehat{D_{i+1}^\ell}[0]}_2 \\&+  \norm{\widehat{D_1^\ell \cdots D_i^\ell}[0] }_2 \sum_{t \ne 0} p^{|t|} \norm{\widehat{D_{i+1}^\ell}[t]}_2\\
\leq& L_p(D_1^\ell \cdots D_i^\ell) \cdot L_p(D_{i+1}^\ell) + L_p(D_1^\ell \cdots D_i^\ell) \lambda(D_{i+1}^\ell) \\&+ \norm{\widehat{D_1^\ell \cdots D_i^\ell}[0] }_2 \cdot L_p(D_{i+1}^\ell)\\
\leq& L_p(D_1^\ell \cdots D_i^\ell) \cdot 1 + \sqrt{2} L_p(D_{i+1}^\ell)\\
\leq& 2i+2.
\end{align*}
The second inequality follows from Equation \ref{eqn:mix-coeff} and the third from Lemma \ref{LemmaLL}.  
Thus, we have that 
$L^k(D^\ell)\leq p^{-k} L_p(D^\ell) \leq 2n \cdot (6000(\ell+1))^k,$ as required
\end{proof}

Now we turn our attention to Lemma \ref{LemmaLL}. We split into two cases: If $\lambda(D)$ is far from $1$ i.e. $\lambda(D) \leq 0.99$, then we need only ensure $L_p(D) \leq 1/100$. This is the `easy case' which proceeds much like the analysis of regular branching programs \cite{ReingoldSteinkeVadhan2013}. 
If $\lambda(D)=1$, then $D$ is trivial -- i.e. $L_p(D)=0$ -- and we are also done. The `hard case' is when $\lambda(D)$ is very close to $1$. i.e. $0.99 \leq \lambda(D) < 1$.

\subsubsection*{Easy Case -- Good Mixing}

We consider the case where $\lambda(D) < 0.99$. We use the following result as a black box.
\begin{lemma}[{\cite[Lemma 4]{BRRY}, \cite[Lemma 3.1]{ReingoldSteinkeVadhan2013}}] \label{LemmaBRRY}
Let $B$ be a length-$n$, width-$w$, ordered, regular branching program. Then $$\sum_{1 \leq i \leq n} \norm{ \widehat{B_{i \cdots n}}[1 \circ 0^{n-i}] }_2 \leq 2w^2.$$
\end{lemma}
The quantity $\norm{ \widehat{B_{i \cdots n}}[1 \circ 0^{n-i}] }_2$ measures the correlation between the $i^\text{th}$ input bit and the final state of the program, which we call the \textbf{weight} of bit $i$. The entry in the $u^\text{th}$ row and $v^\text{th}$ column of $2\widehat{B_{i \cdots n}}[1 \circ 0^{n-i}] $ is the the probability of reaching vertex $v$ in layer $n$ given that we reached vertex $u$ in layer $i-1$ and the $i^\text{th}$ input bit is $0$ minus the same probability given that the $i^\text{th}$ input bit is $1$. Braverman et al.~\cite{BRRY} proved this result for a different measure of weight. Their result was translated into the above Fourier-analytic form by Reingold et al.~\cite{ReingoldSteinkeVadhan2013}.

We can add some non-regular layers to get the following.
\begin{lemma} \label{LemmaBRRYnonregular}
Let $B$ be a length-$n$, width-$w$, ordered branching program with at most $k$ non-regular layers. Then $$\sum_{1 \leq i \leq n} \norm{ \widehat{B_{i \cdots n}}[1 \circ 0^{n-i}] }_2 \leq (2w^2 + 1)\sqrt{w} (k+1).$$
\end{lemma}
\begin{proof}
The proof proceeds by induction on $k$. If $k=0$, the result follows from Lemma \ref{LemmaBRRY}. Suppose the result holds for some $k$ and let $B$ be a length-$n$, width-$w$ ordered branching program with $k+1$ non-regular layers. Let $i^*$ be the index of the first non-regular layer. Then
\begin{align*}
\sum_{1 \leq i \leq n} \norm{ \widehat{B_{i \cdots n}}[1 \circ 0^{n-i}] }_2
=& \sum_{1 \leq i < i^*} \norm{ \widehat{B_{i \cdots (i^*-1)}}[1 \circ 0^{i^*-i-1}] \cdot \widehat{B_{i^* \cdots n}}[0] }_2\\
&+ \norm{ \widehat{B_{i^* \cdots n}}[1 \circ 0^{n-i^*}] }_2\\
&+\sum_{i^* < i \leq n} \norm{ \widehat{B_{i \cdots n}}[1 \circ 0^{n-i}] }_2\\
\leq& \left( \sum_{1 \leq i < i^*} \norm{ \widehat{B_{i \cdots (i^*-1)}}[1 \circ 0^{i^*-i-1}]}_2 \right) \cdot \norm{\widehat{B_{i^* \cdots n}}[0] }_2 \\&+ \sqrt{w} + (2w^2 + 1)\sqrt{w} (k+1)\\
\leq& 2 w^2 \cdot \sqrt{w} + \sqrt{w} + (2w^2 + 1)\sqrt{w} (k+1)\\
\leq& (2w^2 + 1)\sqrt{w} (k+2),\\
\end{align*}
where we use the fact that $\norm{\widehat{B}[s]}_2 \leq \sqrt{w}$ for any $s$ and width-$w$ branching program $B$.
\end{proof}
This gives us the following bound on the Fourier mass.
\begin{theorem}
Let $B$ be a length-$n$, width-$w$, orderd branching program with at most $k$ non-regular layers. Then, for all $k' \in [n]$, $$L^{k'}(B) := \sum_{s \in \{0,1\}^n : |s|=k} \norm{\widehat{B}[s]}_2 \leq \sqrt{w} \cdot ((2w^2 + 1)\sqrt{w} (k+1))^{k'} \leq \sqrt{w} \cdot (3w^{2.5}(k+1))^{k'}.$$
\end{theorem}
This result is proved using Lemma \ref{LemmaBRRYnonregular} analogously to how \cite[Theorem 3.2]{ReingoldSteinkeVadhan2013} is proved using Lemma \ref{LemmaBRRY}. 
\begin{proof} 
We perform an induction on $k'$. If $k'=0$, then there is only one Fourier coefficient to bound---namely, $\widehat{B}[0^n] = \ex{U}{B[U]}$. Since $\ex{U}{B[U]}$ is stochastic, $\norm{\ex{U}{B[U]}}_2 \leq \sqrt{w}$ and the base case follows. Now suppose the bound holds for $k'$ and consider $k'+1$. We split the Fourier coefficients based on where the last $1$ is:
\begin{align*}
\lefteqn{\sum_{s \in \{0,1\}^n : |s|=k'+1} \norm{\widehat{B}[s]}_2}~~~~~~~~~~&\\ =& \sum_{1 \leq i \leq n} \sum_{s \in \{0,1\}^{i-1} : |s|=k'} \norm{\widehat{B}[s \circ 1 \circ 0^{n-i}]}_2\\
=& \sum_{1 \leq i \leq n} \sum_{s \in \{0,1\}^{i-1} : |s|=k'} \norm{\widehat{B_{1 \cdots i-1}}[s] \cdot \widehat{B_{i \cdots n}}[1 \circ 0^{n-i}]}_2~~~\text{(by Lemma \ref{LemmaFourier} (Decomposition))}\\
\leq& \sum_{1 \leq i \leq n} \sum_{s \in \{0,1\}^{i-1} : |s|=k'} \norm{\widehat{B_{1 \cdots i-1}}[s]}_2 \cdot \norm{ \widehat{B_{i \cdots n}}[1 \circ 0^{n-i}]}_2\\
\leq& \sum_{1 \leq i \leq n} \sqrt{w} \cdot ((2w^2 + 1)\sqrt{w} (k+1))^{k'} \cdot \norm{ \widehat{B_{i \cdots n}}[1 \circ 0^{n-i}]}_2~~~\text{(by the induction hypothesis)}\\
\leq& \sqrt{w} \cdot ((2w^2 + 1)\sqrt{w} (k+1))^{k'} \cdot (2w^2 + 1)\sqrt{w} (k+1) ~~~\text{(by Lemma \ref{LemmaBRRYnonregular})}\\
=&\sqrt{w} \cdot ((2w^2 + 1)\sqrt{w} (k+1))^{k'+1},\\
\end{align*}
as required.
\end{proof}

\begin{lemma} \label{LemmaEasy}
Let $D$ be a 3OBP with at most $k$ non-regular layers. If $p \leq 1/6000 (k+1)$, then $L_p(D) \leq 1/100$.
\end{lemma}
\begin{proof}
We have
\begin{align*}
L_p(D) =& \sum_{k' \geq 1} p^{k'} L^{k'}(D)\\
\leq& \sum_{k' \geq 1} p^{k'}  \sqrt{3} ((2\cdot 3^2 + 1)\sqrt{3} (k+1))^{k'}\\
\leq& \sqrt{3} \sum_{k' \geq 1} \left( \frac{19 \sqrt{3} (k+1)}{6000 (k+1)} \right)^{k'}\\
\leq& 1/100.\\
\end{align*}
\end{proof}
It immediately follows that $\lambda(D) + L_p(D) \leq 1$ when $p \leq 1/6000 (k+1)$, assuming $\lambda(D) < 0.99$. This covers the `easy' case of Lemma \ref{LemmaLL}.

\subsubsection*{Hard Case -- Poor Mixing}

Now we consider the case where $\lambda(D) \in [0.99, 1]$.

\begin{lemma} \label{LemmaHard}
Let $D$ be a 3OBP with at most $k$ non-regular layers where the first and last layers of vertices have width 2. Suppose $\lambda(D) \in [0.99, 1]$. If $p\leq 1/(24k+12)$, then $L_p(D) + \lambda(D) \leq 1$.
\end{lemma}

This covers the `hard' case of Lemma \ref{LemmaLL} and, along with Lemma \ref{LemmaEasy} completes the proof of Lemma \ref{LemmaLL}.

Since $D$ has width 2 in the first and last layers, we view $D[x]$ as a $2 \times 2$ matrix. We can write the expectation (which is stochastic) as $$\ex{U}{D[U]} = \left( \begin{array}{cc} 1-\alpha & \alpha \\ \beta & 1-\beta \end{array} \right).$$ We can assume (by permuting rows and columns) that $\lambda(D) = 1 - \alpha - \beta$ and $\alpha, \beta \in [0,1/100]$. 
Now write $$D[x] = \left( \begin{array}{cc} 1-f(x) & f(x) \\ g(x) & 1-g(x) \end{array} \right),$$ where $f,g : \{0,1\}^d \to \{0,1\}$. Then $\alpha = \ex{U}{f(U)}$ and $\beta = \ex{U}{g(U)}$. We can view $D$ has having two \emph{corresponding} start and end states. The probability that, starting in the first start state, we end in the first end state is $1-\alpha \geq 0.99$. Likewise, the probability that, starting in the second start state, we end in the second end state is $1-\beta \geq 0.99$. The function $f$ is computed by starting in the first start state and accepting if we end in the second end state -- that is, we ``cross over''. Likewise, $g$ computes the function telling us whether we will cross over from the second start state to the first end state. Intuitively, there is a low ($1/100$) probability of crossing over, so the program behaves like two disjoint programs.

We will show that $L_p(f) \leq (12k+6) p \alpha$ and $L_p(g) \leq (12k+6) p \beta$ for $p \leq 1/(6k+3)$, from which the result follows by choosing $p$ such that $L_p(f) \leq \alpha/2$ and $L_p(g) \leq \beta/2$.

The plan is as follows.
\begin{itemize}
\item[1.] Show that we can partition the vertices of $D$ into two sets with $O(k)$ edges crossing between the sets such that each layer has at least one vertex in each set. Intuitively, this partitions $D$ into two width-2 branching programs with a few edges going between them.
\item[2.]  Using this partition, show that we can write $f(x) = \sum_s \prod_j f_{s,j}(x_j)$, where each $f_{s,j}$ is a $\{0,1\}$-valued function computed by a \emph{regular} width-2 branching program, the product is over $O(k)$ terms and the $x_j$s are a partition of $x$.
\item[3.] Let $f_s(x) = \prod_j f_{s,j}(x_j)$ and $\alpha_s = \ex{U}{f_s(U)}$. Show that $L_p(f_s) \leq (12k+6) p \alpha_s$ for $p \leq 1/(6k+3)$. Then $$L_p(f) \leq \sum_s L_p(f_s) \leq \sum_s (12k+6) p \alpha_s \leq (12k+6) p \alpha,$$ as required.
\end{itemize}
The same holds for $g$, which gives the result.

\subsubsection*{Step 1.}

\begin{lemma} \label{LemmaPartition}
Let $D$ be a 3OBP with at most $k$ non-regular layers and width-2 in the first and last layers of vertices.  Suppose $\lambda(D) \in [0.99,1]$. Then there is a partition of the vertices of $D$ such that each layer has at least one vertex in each side of the partition and there are at most $2k+1$ layers with an edge that crosses the partition.
\end{lemma}
\begin{proof}
We assign each vertex of $D$ a \textbf{charge}: The two vertices $v_+$ and $v_-$ in the first layer are assigned charges $1$ and $-1$ respectively. Each edge is assigned a charge that is half the charge of the vertex it originates from and the charge of each subsequent vertex is the sum of the charges of the incoming edges. In other words, the charge of a vertex $u$ is the probabililty that a random walk from $v_+$ reaches $u$ minus the probability that a random walk from $v_-$ reaches $u$.

The partition is given by the sign of the charge: Let $Q$ be the set of vertices with positive or zero charge and let $\overline{Q}$ be the set of vertices with negative charge. Now we must prove that there are $O(k)$ edges crossing between $Q$ and $\overline{Q}$.

Define the \textbf{total charge} of a layer to be the sum of the absolute values of the charges in that layer. Clearly the total charge cannot increase from one layer to the next (by the triangle inequality). Moreover, the total charge in the final layer equals $((1-\alpha)-\alpha) + \left( (1-\beta)-\beta \right)=2\lambda(D)$. 


By assumption ($\lambda(D) \geq 0.99$) the total charge decreases by at most $1/50$. The total charge only decreases when an edge crosses the $(Q,\overline{Q})$ partition, as this is when positive and negative charges cancel. In fact, it decreases by precisely the charge of the crossing edge. 

So there is very little charge crossing the partition. However, it is possible that many edges with little charge cross the partition. To preclude this possibility, we also track the \textbf{minimum charge} of each layer, which is the minimum absolute value of a charge of a vertex in that layer.

Now we use the fact that there are at most $k$ non-regular layers in $D$. Call a layer a \textbf{crossing layer} if it contains an edge that crosses the partition i.e. where the signs of the charges of the endpoints of the edge are different. Clearly there are at most $k$ non-regular crossing layers. We need only account for regular crossing layers.

Consider a regular crossing layer. We will argue that the minimum charge must go from `small' to `large'. Then we will argue that in order for the minimum charge to go from large back to small, a non-regular layer is needed. So each such regular crossing layer must have a corresponding non-regular layer. This ensures that there are at most $k+1$ regular crossing layers, as required.

Let $B_i$ be a regular crossing layer. Let $a \leq b \leq c$ be the charges on the left vertices of the layer. Since $a+b+c=0$ and $|a|+|b|+|c| \geq 1.98$, we have that $a \leq -0.49$ and $c \geq 0.49$. The vertices corresponding to $a$ and $c$ cannot have a common neighbour, as otherwise the total charge would decrease by at least $0.2$. Up to permuting vertices, this leaves three possibilities for the layer, which we depict in Figure \ref{fig:charge}.  

\begin{figure}[h] 
\centering
\includegraphics[scale=.5]{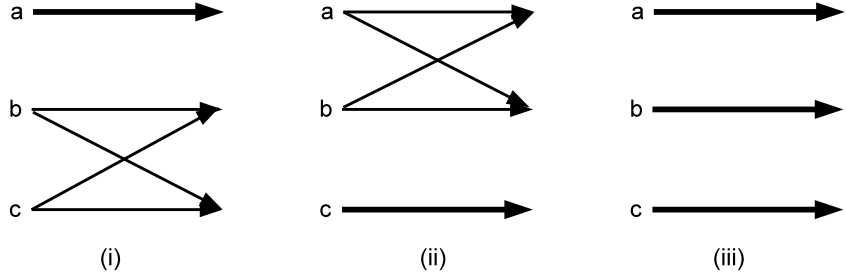}
\caption{The bold arrows indicate double edges.}
\label{fig:charge}
\end{figure}
Possibility (iii) does not have a crossing, so can be ignored. Possibilities (i) and (ii) are essentially the same up to flipping signs. So let's analyse possibility (i).

For there to be a crossing, we must have $b<0$. The total charge then decreases by $|b|$. So $|b| \leq 1/50$. Now $|b|$ is the minimum charge of the vertices on the left. So the minimum charge on the left of a regular crossing layer is at most $1/50$. On the right, the minimum charge is $\min\{|a|,|b+c|/2\} \geq \min\{0.49, (0.49-1/50)/2\} >1/5$. So this layer increases the minimum charge by at least $1/5-1/50 > 0.1$. 

Now we will show that any two regular crossing layers must have a non-regular layer between them. For the sake of contradiction, let $B_i$ and $B_j$ ($i<j$) be two regular crossing layers with no non-regular layers between them. We can assume that there are no crossing layers between $B_i$ and $B_j$: If not, replace $B_j$ with the first crossing layer after $B_i$.

Now consider the vertices in $B_{i+1}$, call them $v_1,v_2$ and $3$.  We assume without loss of generality that one vertex has positive charge (say $v_1$), while $v_2$ and $v_3$ have negative charge.  Because there are no crossing layers, no path from $v_1$ can share a vertex with any path starting from $v_2$ or $v_3$.  Thus, up to permutation on the vertices (determining which vertex is ``isolated''), the first column and first row of the matrix $\ex{U}{B_{i+1,\ldots,j-1[U]}}$ are equal to $(1,0,0)$.  Because every layer of 
$B_{i+1,\ldots,j-1}$ must be regular,  
  up to permuting vertices, we have that $\ex{U}{B_{i+1 \cdots j-1}[U]}$ is of the form $$\left( \begin{array}{ccc} 1 & 0 & 0 \\ 0 & 1 & 0 \\ 0 & 0 & 1 \end{array} \right) ~~~~~~\text{or}~~~~~~ \left( \begin{array}{ccc} 1 & 0 & 0 \\ 0 & 1/2 & 1/2 \\ 0 & 1/2 & 1/2 \end{array} \right).$$ The first possibility cannot decrease the minimum charge at all. The second possibility can only decrease the minimum charge by cancellation: Let $a$, $b$, and $c$ be the charges on the left. Then the charges on the right are $a$, $(b+c)/2$, and $(b+c)/2$. The only way the minimum charge can decrease is if $b$ and $c$ have opposite signs. By symmetry, we can assume that $c \geq -b \geq 0$. Thus the new minimum charge is either $|a|$ or $(c-|b|)/2 = (c+|b|)/2-|b|$. So the minimum charge decreases by at most $|b|$. However, the total charge decreases by $|b|+|c|-|b+c| = c+|b| - (c-|b|) = 2|b|$. Since the total charge can decrease by at most $1/50$, we have $|b|\leq1/100$ and the minimum charge can decrease by at most $1/100$ -- a contradiction, as it must decrease from at least $1/5$ to at most $1/50$.

Thus we have shown that each pair of regular crossing layers has a non-regular layer between then. Since there are at most $k$ non-regular layers, there are at most $2k+1$ crossing layers, as required.
\end{proof}

\subsubsection*{Step 2.}

\begin{lemma} \label{LemmaSumProduct}
Let $D$ be a length-$d$ 3OBP with at most $k$ non-regular layers and width-2 in the first and last layers of vertices.  Suppose $\lambda(D) \geq 0.99$. If $f : \{0,1\}^n \to \{0,1\}$ is the function computed by $D$, then we can write $f(x) = \sum_s \prod_j f_{s,j}(x_j)$, where each $f_{s,j}$ is computed by a regular width-2 ordered branching program and the $x_j$'s are a partition of $x$ into at most $6k+3$ parts.
\end{lemma}
\begin{proof}
Call a layer of edges of $D$ \textbf{critical} if it is either non-regular or it has an edge crossing the partition given by Lemma \ref{LemmaPartition}. Let $\Gamma$ be the set of critical layers. By Lemma \ref{LemmaPartition}, $D$ has at most $3k+1$ critical layers. Between critical layers, $D$ is partitioned into two width-2 regular branching programs. (To be more precise, it is partitioned into a width-2 regular branching program and a width-1 regular branching program.)

Define $\tilde{\Gamma}$ to be the set of `fixings' of edges in $\Gamma$, that is, $s \in \tilde{\Gamma}$ specifies for each $i \in \Gamma$ an edge $s(i)$ in layer $i$ (specifiing one of three states to the left of layer $i$ and the label, which is 0 or 1, of the edge taken). We can think of $s \in \tilde{\Gamma}$ as a function $s : \Gamma \to [3]\times \{0,1\}$.

For $s \in \tilde{\Gamma}$, define $f_s : \{0,1\}^d \to \{0,1\}$ to be the following indicator function: $$f_s(x)=1 \iff f(x)=1 \wedge \text{the path in $D$ given by $x$ uses all the edges in $s$}.$$
Clearly $f(x) = \sum_{s \in \tilde{\Gamma}} f_s(x)$, as each path in $D$ is consistent with exactly one $s$.

Now we claim that each $f_s$ can be written as the conjunction of at most $2|\Gamma|+1$ regular width-2 branching programs. In particular, there is one term for each critical layer and one term for each gap between critical layers.

Let $i_1 < i_2 < \cdots < i_{|\Gamma|}$ be an enumeration of $\Gamma$. (Also define $i_0=0$ and $i_{|\Gamma|+1}=d+1$.) We will write $$f_s(x) = \prod_{j=1}^{2|\Gamma|+1} f_{s,j}(x_j),$$ where the $x_j$s are a partition of $x$ as follows. For $j \in [|\Gamma|+1]$, $x_{2j-1} \in \{0,1\}^{i_j-i_{j-1}-1}$ contains the coordinates from $i_{j-1}+1$ to $i_j-1$ of $x$. For $j \in [|\Gamma|]$, $x_{2j} \in \{0,1\}$ is coordinate $i_j$ of $x$. The function $f_{s,j}$ checks the bits in $x_j$ and is 1 if and only if the path is consistent with $f_s=1$ (assuming the bits outside of $x_j$ are set consistently with $f_s=1$.  Thus $f_{s,2j-1}(x_{2j-1})$ verifies that if started in the state $s(i_{j-1})_1$, the input $x_{2j-1}$ leads $D$ to state $s(i_{j})_1$ in layer $i_{j},$ 
and $f_{s,2j}(x_{2j})$ verifies that $x_{2j} = s(i_{2j})_2$, i.e., that the setting of $x_{2j}$ is the same as the label specified by $s(i_{2j})$.  Note that for $f_{s,j}$ to be satisfied, there is only one correct vertex at each end of the path.  The functions $f_{s,2j}$ are determined by a single literal and can be computed by width-$2$ branching programs, and since the functions $f_{s,2j-1}$ are computed over non-critical layers from a single starting vertex, they are also computed by width-2.        
\end{proof}

\subsubsection*{Step 3.}

We use the following fact about regular width-2 ordered branching programs -- a very simple class of functions.
\begin{lemma} \label{LemmaWidth2}
Let $f : \{0,1\}^n \to \{0,1\}$ be computed by a width-2 regular ordered branching program. Then $\ex{U}{f(U)} \in \{0,1/2,1\}$ and $L_p(f) \leq p/2$ for all $p \in [0,1]$.
\end{lemma}
\begin{proof}
Every layer of a regular width-2 ordered branching program falls into one of three cases: trivial layers (the input bit does not affect the state), a (negated) XOR (flips the state depending on the input bit), or a (negated) dictator (sets the state based on the current input bit regardless of the previous state).  Thus any such branching program is either a constant function, which gives $\ex{U}{f(U)}\in \{0,1\}$ and $L_p(f)=0$, or is a (possibly negated) XOR of a subset of the input bits. In the latter case $f$ has  
one non-trivial coefficient of magnitude $1/2$, which implies $\ex{U}{f(U)}=1/2$ and $L_p(f) \leq p \cdot L(f) \leq p/2$.   
\end{proof}

\begin{lemma} \label{LemmaDecomp}
Let $f :\{0,1\}^n \to \{0,1\}$ be of the form $f(x) = \prod_{j \in [k]} f_j(x_j)$, where the $x_j$s are a partition of $x$ and each $f_j$ is computed by a width-2 ordered regular branching program. Then $L_p(f) \leq 2 k p \cdot \ex{U}{f(U)}$ for any $p \leq 1/k$.
\end{lemma}
\begin{proof}
Define $\alpha_j = \ex{U}{f_j(U)}$. We have $\alpha = \prod_{j \in [k]} \alpha_j$. Now
\begin{align*}
\alpha + L_p(f) 
=& \sum_s p^{|s|} |\widehat{f}[s]|\\
=& \sum_s p^{|s|} \prod_{j \in [k]} |\widehat{f_j}[s_j]|\\
=& \prod_{j \in [k]} \sum_{s_j} p^{|s_j|} |\widehat{f_{j}}[s_j]|\\
=& \prod_{j \in [k]} \left( \alpha_{j} + L_p(f_{j}) \right).
\end{align*}
Since $f_{j}$ is computed by a width-2 regular branching program, $\alpha_{j} \in \{0,1/2,1\}$. If $\alpha_{j}=0$, then $\alpha = 0$ and $L_p(f)=0$, so we can ignore this case. Moreover, if $\alpha_{j} =1$, then $f_{j}=1$ is constant and $L_p(f_{j})=0$, so we can ignore the terms with $\alpha_{j} = 1$.  Let $J = \{j \in [k] : \alpha_{j}=1/2\}$. Thus we are left with 
\begin{align*}
L_p(f) =& \prod_{j \in J} \left( \frac{1}{2} + L_p(f_{j}) \right) - 2^{-|J|}\\
 =& \alpha\cdot \left( \prod_{j \in J} \left( 1 + 2 L_p(f_{j}) \right) - 1 \right)\\
 \leq& \alpha \cdot \left( \prod_{j \in J} \left( 1 + p \right) - 1 \right)\\
\leq& \alpha \cdot \left( (e^{p})^{|J|} -1 \right)\\
\leq& \alpha \cdot 2 p |J|\\
\leq& \alpha 2 p k.\\
\end{align*}
as long as $p|J| \leq 1$.
\end{proof}

\begin{proof}[Proof of Lemma \ref{LemmaHard}]
By Lemma \ref{LemmaSumProduct}, we can write $f(x) = \sum_s \prod_j f_{s,j}(x_j)$. Let $f_s(x) = \prod_j f_{s,j}(x_j)$, where the product is over at most $6k+3$ terms. Then, by Lemma \ref{LemmaDecomp}, $$L_p(f) \leq \sum_s L_p(f_s) \leq \sum_s (12k+6)p \cdot \ex{U}{f_s(U)} = (12k+6)p \alpha,$$ as long as $p \leq 1/(6k+3)$. Likewise $L_p(g) \leq (12k+6)p \beta$ for $p \leq 1/(6k+3)$.

Now $L_p(D) \leq 2L_p(f)+2L_p(g) \leq (24k+12)\cdot p\cdot (\alpha+\beta)$. If $p\leq 1/(24k+12)$, then $L_p(D) + \lambda(D) \leq \alpha + \beta + 1 - \alpha - \beta = 1$, as required.
\end{proof}

\subsection{Bootstrapping} \label{SubSectionBootstrapping}

Proposition \ref{PropositionLow} gives a bound on the Fourier growth of width-3 branching programs of the form $L^k(B) \leq n^2 \cdot O(k)^k$. The $k^k$ term is inconvenient, but can be easily removed by ``bootstrapping'':

The following proposition shows that if we can bound the Fourier mass up to level $O(\log n)$, then we can bound the Fourier mass at all levels.
\begin{proposition} \label{PropositionBootstrapping}
Let $B$ be a length-$n$, ordered branching program such that, for all $i,j,k \in [n]$ with $k \leq 2k^*$ and $i \leq j$, we have $L^k(B_{i \cdots j}) \leq a \cdot b^k$. Suppose $a \cdot n \leq 2^{k^*}$. Then, for all $i,j,k \in [n]$ with $i \leq j$, we have $L^k(B_{i \cdots j}) \leq a \cdot (2b)^k$.
\end{proposition}
The proof is similar to that of Lemma \ref{LemmaWellOrder}.
\begin{proof}
Suppose the proposition is false and fix the smallest $k$ such that the statement does not hold. Clearly $k>2k^*$. Let $k'=k-k^*$. By minimality $L^{k'}(B_{i \cdots j}) \leq  a \cdot (2b)^{k'}$ for all $i \leq j$. Now $$L^k(B_{i \cdots j}) \leq \sum_{\ell=i}^{j+1} L^{k'}(B_{i \cdots \ell-1}) \cdot L^{k^*}(B_{\ell \cdots j}) \leq \sum_{\ell=i}^{j+1} a \cdot (2b)^{k'} \cdot a \cdot b^{k^*} \leq n \cdot a \cdot (2b)^{k'} \cdot a \cdot b^{k^*} = a \cdot (2b)^k \cdot \left( \frac{n a}{2^{k^*}} \right).$$ Since $ n a \leq 2^{k^*}$, we have a contradiction, as we assumed $L^{k^*}(B_{\ell \cdots j}) > a \cdot (2b)^k$.
\end{proof}

Now we combine Propositions \ref{PropositionLow} with \ref{PropositionBootstrapping} to prove Theorem \ref{TheoremLow}

\begin{proof}[Proof of Theorem \ref{TheoremLow}]
Let $B$ be a length-$n$ 3OBP computing $f$ with width 2 in the first and last layers.
By Proposition \ref{PropositionLow}, we have $$L^k(B) \leq 8 n^2 \cdot \left(200000 k \right)^k $$ for any length-$n$ 3OBP $B$ and $k \in [n]$. Since a subprogram of a 3OBP is also a 3OBP, this bound also applies to $L^k(B_{i \cdots j})$ for all $i,j \in [n]$. If we set $k^* = \lceil \log_2 (8n^3) \rceil$, $a = 8n^2$, and $b=200000 \cdot 2k^*$, then the hypotheses of Proposition \ref{PropositionBootstrapping} are satisfied. Thus, for any  $k \in [n]$, we have $$L^k(B) \leq 8n^2 \cdot \left( 2 \cdot 200000 \cdot 2 k^* \right)^k.$$ Since $L^k(f) \leq L^k(B)$, this gives the result.
\end{proof}

\section{Pseudorandom Restrictions} \label{SubSectionPRrestriction}

Our pseudorandom generator repeatedly applys pseudorandom restrictions. For the analysis, we introduce the concept of an averaging restriction as in Gopalan et al.~\cite{GMRTV} and Reingold et al.~\cite{ReingoldSteinkeVadhan2013}, which is subtly different to the restrictions in Section \ref{SubSectionLowOrder}.

\begin{definition}
For $t \in \{0,1\}^n$ and a length-$n$ branching program $B$, let $B|_t$ be the \textbf{(averaging) restriction} of $B$ to $t$---that is, $B|_t : \{0,1\}^n \to \mathbb{R}^{w \times w}$ is a matrix-valued function given by $B|_t [x] := \ex{U}{B[\Set(t,x,U)]}$, where $U$ is uniform on $\{0,1\}^n$. 
\end{definition}

In this section we show that, for a \emph{pseudorandom} $T$ (generated using few random bits), $L(B|_T)$ is small. We will generate $T$ using an almost $O(\log n)$-wise independent distribution:

\begin{definition} \label{DefinitionLimitedIndependence}
A random variable $X$ on $\Omega^n$ is \textbf{$\delta$-almost $k$-wise independent} if, for every\\ $I=\{i_1, i_2, \cdots, i_k\} \subset [n]$ with $|I| = k$, the coordinates $(X_{i_1}, X_{i_2}, \cdots, X_{i_k}) \in \Omega^k$ are $\delta$-close (in statistical distance) to being independent---that is, for all $T \subset \Omega^k$, $$\left|\sum_{x \in T} \left( \pr{X}{(X_{i_1},X_{i_2},\cdots,X_{i_k}) = x} - \prod_{l \in [k]} \pr{X}{X_{i_l} = x_l} \right) \right| \leq \delta.$$ We say that $X$ is \textbf{$k$-wise independent} if it is $0$-almost $k$-wise independent.
\end{definition}

We can sample a random variable $X$ on $\{0,1\}^n$ that is $\delta$-almost $k$-wise independent such that each bit has expectation $p=2^{-d}$ using $O(kd+\log(1/\delta)+d\log(nd))$ random bits \cite[Lemma B.2]{ReingoldSteinkeVadhan2013}. 

The following lemma, proven in essentially the same way as Lemma 5.3 in \cite{ReingoldSteinkeVadhan2013},  tells us that $L(B|_T)$ will be small for $T$ chosen from a $\delta$-almost $k$-wise distribution with appropriate parameters.    
\begin{lemma} \label{TheoremMainLemma}
Let $B$ be a length-$n$, width-$w$, ordered branching program. Let $T$ be a random variable over $\{0,1\}^n$ where each bit has expectation $p$ and the bits are $\delta$-almost $2k$-wise independent. Suppose that, for all $i,j,k' \in [n]$ such that $k \leq k' < 2k$, we have $L^{k'}(B_{i \cdots j}) \leq a\cdot b^{k'}$.  If we set
$p \leq 1/2b$ and $\delta \leq 1/(2b)^{2k},$ then
 $$\pr{T}{L^{\geq k} (B|_T) > 1 } \leq n^4 \cdot \frac{2a}{2^k}.$$
 (Recall that $L^{\geq k}(g) = \sum_{j=k}^n L^j(g)$.)
\end{lemma}
\begin{proof}


Let $k \leq k' < 2k$. We have that for all $i$ and $j$, $$\ex{T}{L^{k'}(B_{i \cdots j}|_T)} = \sum_{s \subset \{i \cdots j\} : |s|=k'} \pr{T}{s \subset T} \norm{\widehat{B_{i \cdots j}}[s]}_2 \leq L^{k'}(B) (p^{k'} + \delta) \leq a b^{k'} \left(\frac{1}{(2b)^{k'}} + \frac{1}{(2b)^{2k}}\right)\leq \frac{2a}{2^k}.$$ 
Applying Markov's inequality and a union bound, we have that for all $\beta > 0$: 
\begin{equation*} 
\pr{T}{\forall 1 \leq i \leq j \leq n ~~ L_2^{k'}(B_{i \cdots j}|_T) \leq \beta} \geq  1-n^2\frac{2a}{2^k \beta}.\end{equation*}

Applying a union bound over values of $k'$ and setting $\beta =1/n$, we obtain: 
$$\pr{T}{\forall k \leq k' < 2k ~ \forall 1 \leq i \leq j \leq n ~~ L^{k'}(B_{i \cdots j}|_T) \leq \frac{1}{n}} \geq 1- n^4 \cdot \frac{2a}{2^k}.$$ 
The result now follows from the following Lemma.
\begin{lemma}[{\cite[Lemma 5.4]{ReingoldSteinkeVadhan2013}}] \label{LemmaWellOrder}
Let $B$ be a length-$n$, ordered branching program and $t \in \{0,1\}^n$. Suppose that, for all $i$, $j$, and $k'$ with $1 \leq i \leq j \leq n$ and $k \leq k' < 2k$, $L_2^{k'}(B_{i \cdots j}|_t) \leq 1/n$. Then, for all $k'' \geq k$ and all $i$ and $j$, $L_2^{k''}(B_{i \cdots j}|_t) \leq 1/n$.
\end{lemma}
\end{proof}

\section{The Pseudorandom Generator} \label{SubSectionPRG}



Our main result Theorem \ref{thm:main-intro} follows from plugging our Fourier growth bound (Theorem \ref{TheoremLow}) into the analysis of \cite{ReingoldSteinkeVadhan2013}. We include the proof and a general statement here for completeness:

\begin{theorem} \label{TheoremPRGformula}
Let $\mathcal{C}$ be a set of ordered branching programs of length at most $n$ and width at most $w$ that is closed under restrictions and subprograms -- that is, if $B \in \mathcal{C}$, then $B|_{t \leftarrow x} \in \mathcal{C}$ for all $t$ and $x$ and $B_{i \cdots j} \in \mathcal{C}$ for all $i$ and $j$. Suppose that, for all $B \in \mathcal{C}$ and $k \in [n]$, we have $L^k(B) \leq a b^k$, where $b \geq 2$. Let $\varepsilon>0$.

Then there exists a pseudorandom generator $G_{a,b,n,\varepsilon} : \{0,1\}^{s_{a,b,n,\varepsilon}} \to \{0,1\}^n$ with seed length $s_{a,b,n,\varepsilon} = O\left(b \cdot \log(b) \cdot \log(n) \cdot \log\left(\frac{abwn}{\varepsilon}\right)\right)$ such that, for any length-$n$, width-$w$, read-once, oblivious (but unordered) branching program $B$ that corresponds to an ordered branching program in $\mathcal{C}$,\footnote{That is, there exists $B' \in \mathcal{C}$ and a permutation of the bits $\pi : \{0,1\}^n \to \{0,1\}^n$ such that $B[x] = B'[\pi(x)]$ for all $x$.} $$\norm{\ex{U_{s_{a,b,n, \varepsilon}}}{B[G_{a,b,n,\varepsilon}(U_{s_{a,b,n, \varepsilon}})]}-\ex{U}{B[U]}}_2 \leq \varepsilon.$$ Moreover, $G_{a,b,n, \varepsilon}$ can be computed in space $O(s_{a,b,n,\varepsilon})$.
\end{theorem}

To prove Theorem \ref{thm:main-intro} we set $\mathcal{C}$ to be the class of all 3OBPs of length at most $n$. Theorem \ref{TheoremLow} gives a bound corresponding to $a = O(n^2)$ and $b = O(\log n)$. This gives the required generator. The statements of Theorems \ref{thm:main-intro} and \ref{TheoremPRGformula} differ in that Theorem \ref{TheoremPRGformula} bounds the error of the pseudorandom generator with respect to a matrix-valued function, while Theorem \ref{thm:main-intro} bounds the error with respect to a $\{0,1\}$-valued function. These statements are equivalent as the $\{0,1\}$-valued function is simply one entry in the matrix-valued function.

The following lemma gives the basis of our pseudorandom generator.

\begin{lemma} \label{LemmaOneStep}
Let $a$, $b$, and $\mathcal{C}$ be as in Theorem \ref{TheoremPRGformula} and $B \in \mathcal{C}$. Let $\varepsilon \in (0,1)$. Let $T$ be a random variable over $\{0,1\}^n$ that is $\delta$-almost $2k$-wise independent and each bit has expectation $p$, where we require $$p \leq 1/(2b),~~~~ k \geq \log_2 \left( 8 a n^4 w / \varepsilon \right), ~~~~\text{and}~~~~\delta \leq 1/(2b)^{2k}.$$ Let $U$ be uniform over $\{0,1\}^n$. Let $X$ be a $\mu$-biased random variable over $\{0,1\}^n$ with $\mu \leq \varepsilon /2ab^k.$  Then $$\norm{\ex{T,X,U}{B[\Set(T,X,U)]} - \ex{U}{B[U]}}_2 \leq \varepsilon.$$
\end{lemma}


\begin{proof}
For a fixed $t \in \{0,1\}^n$, we have
\begin{align*}
\norm{\ex{X,U}{B[\Set(t,X,U)]} - \ex{U}{B[U]}}_2 =& {\norm{\ex{X}{B|_t[X]} - \ex{U}{B[U]}}_2}\\
=& {\norm{\sum_{s \ne 0} \widehat{B|_t}[s] \widehat{X}(s)}_2}\\
\leq& {\sum_{s \ne 0} \norm{\widehat{B|_t}[s]}_2 |\widehat{X}(s)|}\\
\leq& {L(B|_t) \mu},\\
\end{align*}
Conditioning on whether or not $L^{\geq k}(B|_t)>1$, we have
\begin{align*}
\norm{\ex{T,X,U}{B[\Set(T,X,U)]} - \ex{U}{B[U]}}_2 
& \leq \pr{T}{L^{\geq k}(B|_T) > 1} \max_{t} \norm{\ex{X,U}{B[\Set(t,X,U)]} - \ex{U}{B[U]}}_2 \\&+  \pr{T}{L^{\geq k} (B|_T) \leq 1}  \mu \ex{T}{L(B|_T) \mid L^{\geq k}(B|_T) \leq 1}. \\
\end{align*}
We have $$\L^{< k}(B) \leq \sum_{1 \leq k' < k} a b^{k'} = ab \frac{b^{k-1}-1}{b-1} \leq ab^k-1.$$ Thus $\ex{T}{L(B|_T) \mid L^{\geq k}(B|_T) \leq 1} \leq ab^k$.
Lemma \ref{TheoremMainLemma} gives 
$$ 
\pr{T}{L^{\geq k}(B|_T) > 1} \leq n^4 \cdot \frac{2a}{2^k}.
$$
For all $t,x,y$, we have $\norm{B[\Set(t,x,y)] - \ex{U}{B[U]}}_2 \leq 2w$. Thus 
\begin{align*}
\norm{\ex{T,X,U}{B[\Set(T,X,U)]} - \ex{U}{B[U]}}_2 \leq& n^4 \cdot \frac{2a}{2^k} \cdot 2w + 1 \cdot \mu \cdot ab^k\\
\leq& \frac{4 a n^4 w}{8 a n^4 w / \varepsilon} + a b^k \frac{\varepsilon}{2 a b^k}\\
\leq& \varepsilon.
\end{align*}
\end{proof}

Now we use the above results to construct our pseudorandom generator.



The pseudorandom generator is formally defined as follows.
\begin{quote}
\begin{center}\textbf{Algorithm for $G_{a,b,n, \varepsilon} : \{0,1\}^{s_{a,b,n, \varepsilon}} \to \{0,1\}^n$.}\end{center}
\begin{itemize}
\item[Parameters:] $n \in \mathbb{N}$, $\varepsilon>0$.
\item[Input:] A random seed of length $s_{a,b,n, \varepsilon}$.
\item[1.] Compute appropriate values of $p \leq 1/2b$, $\varepsilon' = \varepsilon p / 14w \log_2(n)$, $k \geq \log_2 \left( 8 a n^4 w / \varepsilon' \right)$, $\delta \leq \varepsilon' (p/2)^{2k}$, and $\mu \leq \varepsilon' /2ab^k.$  \footnote{For the purposes of the analysis we assume that $\varepsilon'$, $k$, $p$, $\delta$, and $\mu$ are the same at every level of recursion. So if $G_{a,b,n,w,\varepsilon}$ is being called recursively, use the same values of $\varepsilon'$, $p$, $k$, $\delta$, and $\mu$ as at the previous level of recursion. We pick values within a constant factor of these constraints.}
\item[2.] If $n \leq 320 \cdot \lceil \log_2(1/\varepsilon') \rceil/p$, output $n$ truly random bits and stop.
\item[3.] Sample $T \in \{0,1\}^n$ where each bit has expectation $p$ and the bits are $\delta$-almost $2k$-wise independent.
\item[4.] If $|T|<pn/2$, output $0^n$ and stop.
\item[5.] Recursively sample $\tilde{U} \in \{0,1\}^{\lfloor n(1-p/2) \rfloor}$. i.e. $\tilde{U}=G_{a,b,\lfloor n(1-p/2) \rfloor,\varepsilon}(U)$. 
\item[6.] Sample $X \in \{0,1\}^n$ from a $\mu$-biased distribution.
\item[7.] Output $\mathrm{Select}(T,X,\tilde{U}) \in \{0,1\}^n$.\footnote{Technically, we must pad $\tilde{U}$ with zeros in the locations specified by $T$ (i.e. $\tilde{U}_i = 0$ for $i \in T$) to obtain the right length.}
\end{itemize}
\end{quote}

The analysis of the algorithm proceeds roughly as follows.
\begin{itemize}
\item We have $p = \Theta(1/b)$, $\varepsilon' = \Theta(\varepsilon/wb\log n)$, $k = \Theta(\log(abwn/\varepsilon))$, $\delta=1/b^{\Theta(k)}$, and $\mu=1/b^{\Theta(k)}$.
\item Every time we recurse, $n$ is decreased to $\lfloor n(1-p/2) \rfloor$. After $O(\log(n)/p)$ recursions, $n$ is reduced to $O(1)$. So the maximum recursion depth is $r=O(\log(n)/p) = O( b \log n )$.
\item The probability of failing because $|T|<pn/2$ is small by a Chernoff bound for limited independence. (This requires that $n$ is not too small and, hence, step 2.)
\item The output is pseudorandom, as $$\ex{U}{B[G_{a,b,n,\varepsilon}(U)]} = \ex{T,X,\tilde{U}}{B[\text{Select}(T,X,\tilde{U})]} \approx \ex{T,X,U}{B[\text{Select}(T,X,{U})]} \approx \ex{U}{B[U]}.$$ The first approximate equality holds because we inductively assume that $\tilde{U}$ is pseudorandom. The second approximate equality holds by Lemma \ref{LemmaOneStep}.
\item The total seed length is the seed length needed to sample $X$ and $T$ at each level of recursion and $O(\log(1/\varepsilon')/p) = O(b \log(bwn/\varepsilon))$ truly random bits at the last level. Sampling $X$ requires seed length $O(\log(n/\mu)) = O(k \log b)$ and sampling $T$ requires seed length $O(k \log(1/p) + \log(1/\delta) + \log(1/p) \cdot \log(n \log(1/p))) = O(k \log b)$ so the total seed length is 
$$r \cdot O(k \log b) + O(b \log(bwn/\varepsilon)) =  O\left(b \cdot \log(b) \cdot \log(n) \cdot \log\left(\frac{abwn}{\varepsilon}\right)\right).$$
\end{itemize}
\begin{lemma} \label{LemmaPRGfail}
The probability that $G_{a,b,n,\varepsilon}$ fails at step 4 is bounded by $3\varepsilon'$---that is, $\pr{T}{|T|<pn/2}\leq 3\varepsilon'$. 
\end{lemma}
\begin{proof}
By a Chernoff bound for limited independence (Lemma \ref{LemmaChernoff}), $$\pr{T}{|T| < pn/2} \leq \left( \frac{20k'}{(1/2)^2pn} \right)^{\lfloor k'/2 \rfloor} + 2 \delta\cdot  \left(\frac{n}{(1/2)pn}\right)^{k'},$$ where $k' \leq 2k$ even is arbitrary. Set $k'= 2 \lceil \log_2(1/\varepsilon') \rceil$. Step 2 ensures that 
$n > 160 k'/p$ and our setting of $\delta$ gives that 
$\delta \leq \varepsilon' (p/2)^{k'}.$
Thus we have $$\pr{T}{|T| < pn/2} \leq 2^{-\log_2(1/\varepsilon')} + 2 \varepsilon' \leq 3\varepsilon'.$$
\end{proof}

The following bounds the error of $G_{a,b,n,\varepsilon}$.

\begin{lemma} \label{LemmaPRGerror}
Let $B \in \mathcal{C}$. Then 
$$\norm{\ex{U_{s_{n,\varepsilon}}}{B[G_{n,\varepsilon}(U_{s_{n,\varepsilon}})]}-\ex{U}{B[U]}}_2 \leq 7w r \varepsilon' ,$$ where $r=O(\log(n)/p)$ is the maximum recursion depth of $G_{a,b,n,\varepsilon}$.
\end{lemma}
\begin{proof}
For $0 \leq i < r$, let $n_i$, $T_i$, $X_i$, and $\tilde{U}_i$ be the values of $n$, $T$, $X$, and $\tilde{U}$ at recursion level $i$. We have $n_{i+1}=\lfloor n_i(1-p/2) \rfloor \leq n(1-p/2)^{i+1}$ and $\tilde{U}_{i-1} = \mathrm{Select}(T_i,X_i,\tilde{U}_i)$. Let $\Delta_i$ be the error of the output from the $i^\text{th}$ level of recursion---that is, $$\Delta_i := \max_{B' \in \mathcal{C}} \norm{\ex{T_i,X_i,\tilde{U}_i}{B'[\mathrm{Select}(T_i,X_i,\tilde{U}_i)]}-\ex{U}{B'[U]}}_2.$$

Since the last level of recursion outputs uniform randomness, $\Delta_r=0$. For $0 \leq i<r$, we have, for some $B' \in \mathcal{C}$, 
\begin{align*}
\Delta_i \leq& \norm{\ex{T_i,X_i,\tilde{U}_i}{B'[\text{Select}(T_i,X_i,\tilde{U}_i)]}-\ex{U}{B'[U]}}_2 \cdot \pr{T}{|T|\geq pn/2} \\&+ 2w \cdot \pr{T}{|T|<pn/2}\\
\leq& \norm{\ex{T_i,X_i,\tilde{U}_i}{B'[\text{Select}(T_i,X_i,\tilde{U}_i)]}-\ex{T_i,X_i,{U}}{B'[\text{Select}(T_i,X_i,{U})]}}_2 \\&+ \norm{\ex{T_i,X_i,{U}}{B'[\text{Select}(T_i,X_i,{U})]}-\ex{U}{B'[U]}}_2 \\&+2w\cdot\pr{T}{|T|<pn/2}\\
\end{align*}
By Lemma \ref{LemmaOneStep}, $$\norm{\ex{T_i,X_i,{U}}{B'[\text{Select}(T_i,X_i,{U})]}-\ex{U}{B'[U]}}_2 \leq  \varepsilon'.$$ By Lemma \ref{LemmaPRGfail}, $$\pr{T}{|T|<pn/2} \leq 3\varepsilon'.$$
We claim that $$\norm{\ex{T_i,X_i,\tilde{U}_i}{B'[\text{Select}(T_i,X_i,\tilde{U}_i)]}-\ex{T_i,X_i,{U}}{B'[\text{Select}(T_i,X_i,{U})]}}_2 \leq \Delta_{i+1}.$$ Before we prove the claim, we complete the proof: This gives $\Delta_i \leq \Delta_{i+1} + \varepsilon' + 2w \cdot 3\varepsilon'$. It follows that $\Delta_0 \leq 7w r \varepsilon'$, as required. 

To prove the claim,  consider \emph{any} fixed $T_i=t$ and $X_i=x$. We have $$\norm{\ex{\tilde{U}_i}{B'[\text{Select}(t,x,\tilde{U}_i)]}-\ex{U}{B'[\text{Select}(t,x,{U})]}}_2 \leq \Delta_{i+1}.$$

Consider $\overline{B}_{x,t}[y] := B'[\text{Select}(t,x,y)]$ as a function of $y \in \{0,1\}^{n_i-|t|}$. Then $\overline{B}_{x,t}$ is a width-3 read-once oblivious branching program of length-$(n_i-|t|)$.

We inductively know that $\tilde{U}_i $ is pseudorandom for $\overline{B}_{x,t}$---that is, $\norm{\ex{\tilde{U}_i}{\overline{B}_{x,t}[\tilde{U}_i]} - \ex{U}{\overline{B}_{x,t}[U]}}_2 \leq \Delta_{i+1}$. Thus $$\norm{\ex{\tilde{U}_i}{B'[\text{Select}(t,x,\tilde{U}_i)]}-\ex{U}{B'[\text{Select}(t,x,{U})]}}_2 = \norm{\ex{\tilde{U}_i}{\overline{B}_{x,t}[\tilde{U}_i]}-\ex{U}{\overline{B}_{x,t}[U]}}_2 \leq \Delta_{i+1},$$ as required.
\end{proof}

\begin{proof}[Proof of Theorem \ref{TheoremPRGformula}]
Since  $\varepsilon' \leq \varepsilon/(7w r) $, Lemma \ref{LemmaPRGerror} implies that $G_{a,b,n,\varepsilon}$ has error at most $\varepsilon$. 
The seed length is $$s_{a,b,n,\varepsilon}=O\left(b \cdot \log(b) \cdot \log(n) \cdot \log\left(\frac{abwn}{\varepsilon}\right)\right)$$ as required.
\end{proof}

\section{Further Work} \label{SectionConclusion}

Our results hinge on the fact that ``mixing'' is well-understood for regular branching programs \cite{BRRY,ReingoldSteinkeVadhan2013,KNP,De,Steinke12} and for (non-regular) width-2 branching programs \cite{BogdanovDvVeYe09}. We are able to use random restrictions to reduce from width 3 to width 2 (Section \ref{subsect:width-reduction}), where we can exploit our understanding of mixing (Section \ref{SubSectionMixing}). Indeed, this understanding underpins most results for these restricted models of branching programs. 

What about width 4 and beyond? Using a random restriction we can reduce analysing width 4 to ``almost'' width 3 -- that is, Proposition \ref{PropositionPart1} generalises. Unfortunately, the reduction does not give a true width-3 branching program and thus we cannot repeat the reduction to width 2. Moreover, we have a poor understanding of mixing for non-regular width-3 branching programs, which means we cannot use the same techniques that have worked for width-2 branching programs.

Our results provide some understanding of mixing in width-3. We hope this understanding can be developed further and will lead to proving Conjecture \ref{conj:fouriergrowth} and other results.

The biggest obstacle to extending our techniques to $w>3$ is Lemma \ref{LemmaLL}.
The problem is that the 
parameter $\lambda(D)$ is no longer a useful measure of mixing for width-3 and above. In particular, $\lambda(D)>1$ is possible if $\ex{U}{D[U]}$ is a $3 \times 3$ matrix.
To extend our techniques, we need a better notion of mixing. Using $\lambda(D)$ is useful for regular branching programs (it equals the second eigenvalue for regular programs), but is of limited use for non-regular branching programs. Our proof uses a different notion of mixing -- collisions: To prove Proposition \ref{PropositionPart1}, we used the fact that a random restriction of a non-regular layer will with probability at least 1/2 result in the width of the right side of the layer being reduced. This is a form of mixing, but it is not captured by $\lambda$. Ideally, we want a notion of mixing that captures both $\lambda$ and width-reduction under restrictions.

Our proofs combine the techniques of Braverman et al.~\cite{BRRY} and those of Brody and Verbin \cite{BrodyVerbin} and Steinberger \cite{Steinberger}. We would like to combine them more cleanly -- presently the proof is split into two parts (Proposition \ref{PropositionPart1} and Lemma \ref{LemmaLL}). This would likely involve developing a deeper understanding of the notion of mixing.

Our seed length $\tilde{O}(\log^3 n)$ is far from the optimal $O(\log n)$. Further improvement would require some new techniques:

We could potentially relax our notion of Fourier growth to achieve better results. Rather than bounding $L^k(f)$, it suffices to bound $L^k(g)$, where $g$ \emph{approximates} $f$:

\begin{proposition}[{\cite[Proposition 2.6]{DETT}}] \label{PropositionSandwich}
Let $f, f_+, f_- : \{0,1\}^n \to \RR$ satisfy $f_-(x) \leq f(x) \leq f_+(x)$ for all $x$ and $\ex{U}{f_+(U) - f_-(U)} \leq \delta$. Then any $\varepsilon$-biased distribution $X$ gives $$\left| \ex{X}{f(X)} - \ex{U}{f(U)} \right| \leq \delta  + \varepsilon \cdot \max \left\{ L(f_+), L(f_-) \right\}.$$
\end{proposition}

The functions $f_+$ and $f_-$ are called sandwiching polynomials for $f$. This notion of sandwiching is in fact a tight characterisation of small bias \cite[Proposition 2.7]{DETT}. That is, any function $f$ fooled by all small bias generators has sandwiching polynomials satisfying the hypotheses of Proposition \ref{PropositionSandwich}.

Gopalan et al.~\cite{GMRTV} use sandwiching polynomials in the analysis of their generator for CNFs. This allows them to set a constant fraction of the bits at each level of recursion ($p = \Omega(1)$), while we set a $1/O(\log n)$ fraction at each level. We would like to similarly exploit sandwiching polynomials for branching programs to improve the seed length of the generator. 

A further avenue for improvement would be to modify the generator construction to have $\Theta(1/p)$ levels of recursion, rather than $\Theta(\log(n)/p)$. This would require a significantly different analysis.


\begin{small}
\bibliographystyle{plain}
\bibliography{fourier}
\end{small}

\appendix

\section{Chernoff Bound for Limited Independence}

\begin{lemma}[Chernoff Bound for Limited Independence] \label{LemmaChernoff}
Let $X_1 \cdots X_\ell$ be $\delta$-almost $k$-wise independent random variables with $X_i \in \{0,1\}$ for all $i$. Set $X = \sum_i X_i$ and $\mu = \sum_i \mu_i = \sum_i \ex{X}{X_i}$, and suppose $\mu_i \leq 1/2$ for all $i$. If $k \leq \mu /10$ is even, then, for all $\alpha \in (0,1)$, $$\pr{X}{\left|X - \mu \right| \geq \alpha \mu } \leq \left( \frac{20k}{\alpha^2 \mu } \right)^{\lfloor k/2 \rfloor} + 2\delta \cdot \left(\frac{\ell}{\alpha\mu}\right)^k.$$
\end{lemma}

The following proof is based on \cite[Theorem 4]{SSS}. The only difference is that we extend to almost $k$-wise independence from $k$-wise independence.

\begin{proof}
Assume, without loss of generality, that $k$ is even. It is well-known \cite{SSS,JelaniMO,BRchernoff} that, if the $X_i$s are fully independent, then $$\ex{X_i}{(X-\mu)^k} \leq (20k \mu)^{k/2}.$$ This also holds when the $X_i$s are only $k$-wise independent, as $(X-\mu)^k$ is a degree-$k$ polynomial in the $X_i$s. Here the $X_i$s are $\delta$-almost $k$-wise independent, which gives $$\ex{X_i}{(X-\mu)^k} \leq (20k \mu)^{k/2} + 2 \delta \ell^k.$$
Thus we can apply Markov's inequality to obtain the result:
$$\pr{X}{\left|X - \mu \right| \geq \alpha \mu } = \pr{X}{(X - \mu )^k \geq (\alpha \mu)^k } \leq \frac{\ex{X_i}{(X-\mu)^k}}{(\alpha \mu)^k} \leq \left( \frac{20 k \mu}{\alpha^2 \mu^2} \right)^{k/2} + 2 \delta \left(\frac{\ell}{\alpha \mu} \right)^k.$$
\end{proof}

\section{First-Order Fourier Coefficients of Branching Programs} \label{AppendixFirstOrder}

\begin{theorem}
Let $B$ be a width-$w$, length-$n$, read-once, oblivious branching program. Then $$\sum_{i \in [n]} \norm{\widehat{B}[\{i\}]}_2 \leq O(\log n)^{w-2}.$$
\end{theorem}

The proof is similar to the proof of the Coin Theorem by Steinberger \cite{Steinberger}. The main difference is that we need a new proof of the collision lemma:

We call a layer $B_i$ of a branching program \textbf{trivial} if $L(B_i)=0$ and nontrivial otherwise. We say that a layer $B_i$ has a \textbf{collision} if there exist two edges with the same label and the same endpoint, but different start points. All non-permutation layers have a collision. If there is a collision in layer $i$, then with probability at least $1/2$ a random restriction of layer $i$ reduces the width.

\begin{lemma}[Collision Lemma] \label{lemcollision}
Let $f: \{0,1\}^n \to \{0,1\}$ be a function computed by a width-$w$ ordered branching program $B$. Then there exists a function $g$ computed by a width-$w$ ordered branching program $B'$ such that every nontrivial layer of $B'$ has a collision and $$\sum_i |\widehat{f}[\{i\}]|  \leq \sum_i |\widehat{g}[\{i\}]|.$$
\end{lemma}
\begin{proof}
We will construct $B'$ by flipping edge labels in $B$.

Begin by ordering the vertices in each layer by their acceptance probability -- that is, the probability that $f(U)=1$ conditioned on passing through that vertex. (If two vertices have the same acceptance probability order them arbitrarily.) We will flip the edge labels such that for every vertex the 0-edge leads to a higher-ranked vertex than the 1-edge (or they lead to the same vertex).

Note that flipping the edges in layer $i$ only affects the corresponding Fourier coefficient. Thus we need only show that flipping the edges in layer $i$ does not decrease $|\widehat{f}[i]|$.

Fix a layer $i$. For a vertex $u$ on the left of layer $i$ of $B$, let $u_0$ and $u_1$ be the vertices led to by the 0- and 1-edges respectively and let $p_u$ be the probability that a random walk in $B$ reaches $u$. For a vertex $v$ on the right of layer $i$ of $B$, let $q_v$ be the acceptance probability of $B$ under uniform input conditioned on passing through vertex $v$. Then $$\widehat{f}[\{i\}] = \frac12 \sum_u p_u (q_{u_0} - q_{u_1}).$$ Flipping edge labels corresponds to flipping the signs of the terms in the above sum. Clearly, $|\widehat{f}[\{i\}]|$ is maximised if every term has the same sign. Our choice of flips ensures this is the case, as $q_{u_0} \geq q_{u_1}$.

Every nontrivial layer of $B$ must have a collision, as a result of the ordering of edge labels: Consider a layer $i$ and let $v$ be the highest ranked vertex on its right such that the incoming edges are from different vertices on the left. Suppose, for the sake of contradiction, that the incoming edges have different labels. Pick the edge labelled 1 and let u be its start point. Let $u_0$ be the vertex reached from $u$ by the edge labelled $0$. Then, by our choice of labels $u_0$ is ranked higher than $u$ -- a contradiction, as $u$ is the highest ranked vertex with distinct incoming edges.
\end{proof}

Define $\xi(n,w)$ to be the maximal first-order Fourier mass of any function computed by a length-$n$, width-$w$, ordered branching program. We follow the structure of Steinberger's proof \cite{Steinberger} to bound $\xi$.

\begin{lemma}
For all $n$ and $w \geq 3$, $\xi(n,w) \leq (2 + 2 \log_2(n)) \cdot (\xi(n,w-1)+1)$. 
\end{lemma}
\begin{proof}
Let $B$ be a length-$n$, width-$w$, ordered branching program that maximises the first-order Fourier mass of the function $f$ it computes. By Lemma \ref{lemcollision}, we may assume that every nontrivial layer of $B$ has a collision. We may assume that there are no trivial layers: otherwise we can remove them without affecting the Fourier mass. 

Let $m = \lceil 1 + 2 \log_2 n \rceil$. Split the first-order Fourier coefficients into $m$ groups of the form $$G_{i'} = \{ i \in [n] : i ~\mathrm{mod}~ m = i' \}.$$ We bound the first-order Fourier mass of each group separately and sum them together. i.e. $\sum_{i \in [n]} |\widehat{f}[\{i\}]| = \sum_{i' \in [m]} \sum_{i \in G_{i'}} |\widehat{f}[\{i\}]|$. Fix one group $G = G_{i'}$.

We apply a random restriction to $B$ to obtain the function $f|_{\overline{G} \leftarrow U}$ computed by the branching program $B|_{\overline{G} \leftarrow U}$. We have $$\sum_{i \in G} |\widehat{f}[\{i\}] \leq \ex{U}{\sum_{i \in G} |\widehat{f|_{\overline{G} \leftarrow U}}[\{i\}]|}.$$ So if suffices to bound the first-order Fourier mass of $f|_{\overline{G} \leftarrow U}$.

We claim that $B|_{\overline{G} \leftarrow U}$ is a width-$(w-1)$, ordered branching program with probability at least $1 - n \cdot 2^{1-m}$. This implies that
$$\ex{U}{\sum_{i \in G} |\widehat{f|_{\overline{G} \leftarrow U}}[\{i\}]|} \leq \xi(n,w-1) + n \cdot 2^{1-w} \cdot n \leq \xi(n,w-1) + 1.$$
Thus $$\sum_{i \in [n]} |\widehat{f}[\{i\}]| \leq \sum_{i' \in [m]} \ex{U}{\sum_{i \in G_{i'}} |\widehat{f|_{\overline{G_{i'}} \leftarrow U}}[\{i\}]|} \leq \sum_{i' \in [m]} \xi(n,w-1) + 1 \leq m (\xi(n,w-1) + 1),$$ as required.

Now to prove the claim: Fix an unrestricted layer $i$ of $B|_{\overline{G} \leftarrow U}$ other than the last layer (which can always be assumed to have width 2 anyway). Layer $i$ is followed by $m-1$ restricted layers. With probability at least $1-2^{1-m}$ at least one of these layers will contain a collision, thus reducing the number of vertices on the right of layer $i$. A union bound gives the required probability.
\end{proof}

\begin{lemma}
$\xi(n,2) \leq 10$.
\end{lemma}

Solving the recurrance for $\xi$ gives $\xi(n,w) \leq O(\log n)^{w-2}$, as required.

\omitted{
\section{Bootstrapping}

The following proposition shows that if we can bound the Fourier mass up to level $O(\log n)$, then we can bound the Fourier mass at all levels.
\begin{proposition} \label{PropositionBootstrapping}
Let $B$ be a length-$n$, ordered branching program such that, for all $i,j,k \in [n]$ with $k \leq 2k^*$ and $i \leq j$, we have $L^k(B_{i \cdots j}) \leq a \cdot b^k$. Suppose $a \cdot n \leq 2^{k^*}$. Then, for all $i,j,k \in [n]$ with $i \leq j$, we have $L^k(B_{i \cdots j}) \leq a \cdot (2b)^k$.
\end{proposition}
The proof is similar to that of Lemma \ref{LemmaWellOrder}.
\begin{proof}
Suppose the proposition is false and fix the smallest $k$ such that the statement does not hold. Clearly $k>2k^*$. Let $k'=k-k^*$. By minimality $L^{k'}(B_{i \cdots j}) \leq  a \cdot (2b)^{k'}$ for all $i \leq j$. Now $$L^k(B_{i \cdots j}) \leq \sum_{\ell=i}^{j+1} L^{k'}(B_{i \cdots \ell-1}) \cdot L^{k^*}(B_{\ell \cdots j}) \leq \sum_{\ell=i}^{j+1} a \cdot (2b)^{k'} \cdot a \cdot b^{k^*} \leq n \cdot a \cdot (2b)^{k'} \cdot a \cdot b^{k^*} = a \cdot (2b)^k \cdot \left( \frac{n a}{2^{k^*}} \right).$$ Since $ n a \leq 2^{k^*}$, we have a contradiction, as we assumed $L^{k^*}(B_{\ell \cdots j}) > a \cdot (2b)^k$.
\end{proof}
}

\section{Optimality of Result} \label{AppendixOptimal}

The following result shows that Theorem \ref{TheoremLow} is close to optimal.


\begin{proposition} \label{PropositionTight}
There exists an infinite family of functions $f_n: \{0,1\}^n \to \{0,1\}$ and that are computed by 3OBPs such that the following holds. Let $q : \mathbb{N} \to \mathbb{R}$ be an increasing function with $20 \leq q(n) \leq \exp(\exp(o(\sqrt{\log n})))$. For all sufficiently large $n$, there exists $k \in [n]$ such that $$L^k(f_n) > q(n) \cdot \left(\frac{\log n}{20 \log \log q(n)}\right)^k.$$.
\end{proposition}
Our main result shows that $L^k(f_n) \leq \poly(n) \cdot (O(\log n))^k$. Setting $q(n) = \poly(n)$, this proposition shows that the base $O(\log n)$ cannot be improved by more than a $\log \log n$ factor. The $\log \log n$ factor comes from the fact that we allow a polynomial factor $q(n)$ in the bound. If we demand $q(n) = O(1)$, the base $O(\log n)$ is optimal.
\begin{proof}
Let $n = m \cdot 2^m$, where $m$ is an integer. Define $f_{n} : \{0,1\}^{n} \to \{0,1\}$ by $$f_{n}(x) = \prod_{i \in [2^m]} \left( 1 - \prod_{j \in [m]} x_{i,j} \right),$$ where we view $x \in \{0,1\}^n$ as a $2^m \times m$ matrix $x \in \{0,1\}^{2^m \times m}$. This is (up to a negation) the Tribes function \cite{tribes}. This function can be computed by a 3OBP. Now we show that it has large Fourier growth.

The Fourier coefficients of $f_{n}$ are given as follows. For $s \subset [n]$ (which we identify with $s \in \{0,1\}^{2^m \times m}$),
\begin{align*}
\widehat{f_{n}}[s] =& \ex{U}{f(U)\chi_s(U)}\\
=& \ex{U}{ \prod_{i \in [2^m]} \left( 1 - \prod_{j \in [m]} U_{i,j} \right) \chi_{s_i}(U_i)}\\
=&\prod_{i \in [2^m]}  \ex{U}{\chi_{s_i}(U_i)  - \prod_{j \in [m]} U_{i,j}  \chi_{s_i}(U_i)}\\
=&\prod_{i \in [2^m]}  \left( \mathbb{I}(s_i = 0) - 2^{-m} (-1)^{|s_i|} \right),
\end{align*}
where $s_i = (s_{i,1},s_{i,2},\cdots,s_{i,m})$.
The damped Fourier mass is also easy to compute. For $p \in (0,1)$,
\begin{align*}
L_p(f_{n}) + |\widehat{f_{n}}[0]|=& \sum_{s \in \{0,1\}^{2^m \times m}} p^{|s|} |\widehat{f_{n}}[s]|\\
=& \sum_{s \in \{0,1\}^{2^m \times m}} \prod_{i \in [2^m]}  p^{|s_i|} \left| \mathbb{I}(s_i = 0) - 2^{-m} (-1)^{|s_i|} \right|\\
=& \prod_{i \in [2^m]} \sum_{s_i \in \{0,1\}^m}   p^{|s_i|} \left| \mathbb{I}(s_i = 0) - 2^{-m} (-1)^{|s_i|} \right|\\
=& \prod_{i \in [2^m]} \left( 1-2^{-m} + \sum_{s_i \ne 0} p^{|s_i|} 2^{-m} \right)\\
=& \prod_{i \in [2^m]} \left( 1- 2^{-m} + 2^{-m} (1+p)^m - 2^{-m} \right)\\
=& \left( 1 + \frac{(1+p)^m - 2}{2^m} \right)^{2^m}.
\end{align*}
Set $p=(1+\log (3 + \log q(n))) / m$. We have $$(1+p)^m = \left( 1 + \frac{1 + \log (3 + \log q(n))}{m}\right)^m = e^{1 + \log (3 + \log q(n))} (1 - o(1)) \geq 3 + \log q(n)$$ for sufficiently large $m$. Thus, for sufficiently large $m$, $$L_p(f_{n}) \geq \left( 1 + \frac{(1+p)^m - 2}{2^m} \right)^{2^m} -1 \geq \left( 1 + \frac{1 + \log q(n)}{2^m} \right)^{2^m}-1 = e^{1 + \log q(n)}(1 - o(1)) -1 \geq q(n).$$

Suppose for the sake of contradiction that $L^k(f_n) \leq q(n) \cdot (\log n / 20 \log \log q(n))^k$ for all $k \in [n]$. We have $$L_p(f_{n}) = \sum_{k \in [n]} p^k L^k(f_{n}) \leq \sum_{k \in [n]} \left(\frac{1 + \log (3 + \log q(n))}{m}\right)^k \cdot q(n) \cdot \left(\frac{\log n}{20 \log \log q(n)}\right)^k \leq \sum_{k \in [n]} q(n) 2^{-k} < q(n),$$ which is a contradiction.
\end{proof}

A more careful analysis gives the following bound.

\begin{proposition}
There exists an infinte family of functions $f_n : \{0,1\}^n \to \{0,1\}$ that are computed by 3OBPs such that, for all $k \in [n]$, $$L^k(f) \geq \Omega\left(\frac{\log n}{\log k}\right)^k.$$
\end{proposition}
\begin{proof}
Let $n$, $m$, and $f_n$ be as in the proof of Proposition \ref{PropositionTight}. For $s \in \{0,1\}^{2^m \times m}$, denote $$\ell(s) = |\{i \in [2^m] : s_i \ne 0 \}| = |\{ i \in [2^m] : \exists j \in [m] ~~ s_{i,j}=1\}|.$$ Then, for all $s \in \{0,1\}^{n \times m}$, we have $$|\widehat{f_{n}}[s]| = \prod_{i \in [2^m]}  \left| \mathbb{I}(s_i = 0) - 2^{-m} (-1)^{|s_i|} \right|= (1-2^{-m})^{2^m-\ell(s)} \cdot (2^{-m})^{\ell(s)}.$$
Fix $\ell$ with $k/\ell \leq m$. Set $h=\lfloor k/\ell \rfloor$. Choose $i,j \geq 0$ with $i+j=\ell$ and $ih+j(h+1) = k$. Then 
\begin{align*}
L^k(f_n) \geq& \sum_{|s|=k \wedge \ell(s)=\ell} |\widehat{f_{n}}[s]|\\
\geq& {2^m \choose \ell} {m \choose h}^i {m \choose h+1}^j \cdot (1-2^{-m})^{2^m-\ell} \cdot (2^{-m})^{\ell}\\
\geq& \left( \frac{2^m}{\ell} \right)^\ell \left(\frac{m}{h}\right)^{hi}  \left(\frac{m}{h+1}\right)^{(h+1)j} \cdot \left(1-\frac{1}{2^m}\right)^{2^m} \cdot \left(\frac{1}{2^m}\right)^{\ell}\\
\geq& \frac{1}{4} \left( \frac{1}{\ell} \right)^\ell \left(\frac{m}{h}\right)^{hi}  \left(\frac{m}{h+1}\right)^{(h+1)j}\\
\geq& \frac{1}{4} \left( \frac{1}{\ell} \right)^\ell \left(\frac{m}{h+1}\right)^{hi+(h+1)j}\\
\geq& \frac{1}{4} \left( \frac{1}{\ell} \right)^\ell \cdot \left( \frac{m}{k/\ell+1}\right)^k.
\end{align*}
Suppose $k \leq 2^{m-1}$. Setting $\ell = \lceil k/\log_2 (2k) \rceil$, we have
\begin{align*}
L^k(f_n) \geq& \frac{1}{4} \left( \frac{1}{\ell} \right)^\ell \left(\frac{m}{k/\ell + 1}\right)^{k}\\
\geq& \frac{1}{4} \cdot \frac{1}{2^{2k+2}}  \cdot \left( \frac{m}{\log_2 k + 2}\right)^k,
\end{align*}
as $$\log_2(\ell^\ell) = \ell \log_2 \ell < \frac{k + \log_2 k}{\log_2 k} \log_2 (k + \log_2 k) \leq 2k+2.$$
Since $m = \Theta(\log n)$, this gives the result for $k \leq 2^{m-1}$. If $k > 2^{m-1}$, then $\log k = \Theta(\log n)$ and the result is trivial.
\end{proof}

\end{document}